\documentclass[a4paper,10.9pt,reqno]{article}

\usepackage[T1]{fontenc}			
\usepackage[utf8]{inputenc}
\usepackage[english]{babel}
\usepackage{amsmath, amsfonts, amssymb, amsthm, amscd, mathtools}
\usepackage{geometry}
\usepackage[round]{natbib}	
\usepackage{bbm, bm}
\usepackage{xspace} 
\usepackage{enumerate, enumitem}
\usepackage{hyperref}
\usepackage{xcolor}
\usepackage{comment}
\usepackage{eurosym}
\usepackage{graphicx}

\newcommand{\bbF}{{\ensuremath{\mathbb F}} }

\newcommand{\cB}{{\ensuremath{\mathcal B}} }
\newcommand{\cC}{{\ensuremath{\mathcal C}} }

\newcommand{\cE}{{\ensuremath{\mathcal E}} }
\newcommand{\cF}{{\ensuremath{\mathcal F}} }

\newcommand{\cH}{{\ensuremath{\mathcal H}} }

\newcommand{\cL}{{\ensuremath{\mathcal L}} }

\newcommand{\cO}{{\ensuremath{\mathcal O}} }
\newcommand{\cP}{{\ensuremath{\mathcal P}} }

\newcommand{\cS}{{\ensuremath{\mathcal S}} }

\newcommand{\bfJ}{{\ensuremath{\mathbf J}} }

\newcommand{\dd}{{\ensuremath{\mathrm d}} }

\newcommand{\dC}{{\ensuremath{\mathrm C}} }

\newcommand{\R}{\mathbb{R}}

\renewcommand{\P}{\mathbb{P}}
\newcommand{\E}{\mathbb{E}}

\newcommand{\ind}{\ensuremath{\mathbf{1}}}

\DeclarePairedDelimiterX{\inprod}[2]{\langle}{\rangle}{#1, #2}

\newcommand{\changeepsilon}
{
	\let\temp\epsilon
	\let\epsilon\varepsilon
	\let\varepsilon\temp
}

\newcommand{\changetheta}
{
	\let\temp\theta
	\let\theta\vartheta
	\let\vartheta\temp
}
\newcommand{\changephi}
{
	\let\temp\phi
	\let\phi\varphi
	\let\varphi\temp
}

\theoremstyle{plain}
\newtheorem{theorem}{Theorem}[section]
\newtheorem*{theorem*}{Theorem}
\newtheorem{lemma}[theorem]{Lemma}
\newtheorem*{lemma*}{Lemma}
\newtheorem{proposition}[theorem]{Proposition}
\newtheorem*{proposition*}{Proposition}

\theoremstyle{definition}
\newtheorem{definition}{Definition}[section]
\newtheorem{assumption}{Assumption}[section]

\theoremstyle{remark}

\newtheorem*{remark*}{Remark}

\definecolor{darkviolet}{rgb}{0.58, 0.0, 0.83}

\definecolor{gre}{rgb}{0.03,0.50,0.03}

\numberwithin{equation}{section}

\hypersetup{
	colorlinks,
	citecolor=blue,
	filecolor=blue,
	linkcolor=blue,
	urlcolor=blue
}

\makeatletter
\def\@setthanks{\vspace{-\baselineskip}\def\thanks##1{\@par##1\@addpunct.}\thankses}
\makeatother

\graphicspath{{./figs/}}

\begin{document}
	\title{A Stationary Mean-Field Model \\ for Carbon Market Stability Reserve}
	\author{
		René Aïd\thanks{Université Paris-Dauphine -- PSL Research University, Department of Economics, UMR CNRS 8007. René Aïd thanks the financial support of the {\em Finance and Sustainable Development EDF-CA CIB Chair}, the {\em Finance for Energy Market} Research Initiative, the French ANR PEPR Math-Vives project MIRTE ANR-23-EXMA-0011.}  \quad 
		Federico~Cannerozzi\thanks{Bielefeld University, Center for Mathematical Economics (IMW), Bielefeld (Germany). federico.cannerozzi@uni-bielefeld.de } \quad 
		Giorgio~Ferrari\thanks{ Bielefeld University, Center for Mathematical Economics (IMW), Bielefeld (Germany). giorgio.ferrari@uni-bielefeld.de}
	}
	\date{\today}
	\changephi
	\changeepsilon
	\allowdisplaybreaks
	
	\maketitle
	
	\begin{abstract}
		We formulate the design of a market stability reserve as a stationary mean-field singular control problem. The allowance balance follows a regime-switching diffusion, while the regulator injects or withdraws allowances through bounded-variation controls. The objective trades off the average allowance balance, its stationary dispersion, deviations from a long-run target, and intervention costs.
		The optimal policy is a regime-dependent band rule characterized by four reflecting thresholds. Its computation reduces to a scalar fixed-point problem combining a coupled variational inequality with the invariant distribution of the reflected process. A calibration to the EU ETS yields thresholds of the same order of magnitude as the existing Market Stability Reserve rule. Comparative statics show that volatility is the main determinant of the intervention band, whereas the frequency and duration of crisis regimes have limited effects. A simpler two-threshold policy based on average regime dynamics closely reproduces the optimal regime-dependent policy in terms of liquidity and intervention volumes.
	\end{abstract}
	
	
	\smallskip
	
	\noindent \textbf{AMS 2020}: 93E20; 49N80; 91B76; 91B70. \qquad \textbf{JEL}: Q58; Q54; C61.
	
	\tableofcontents
	
	\section{Introduction}
	
	To decarbonize its economy, the European Union has developed a broad set of transition instruments, combining direct support for low-carbon technologies, renewable energy subsidies, regulatory standards, and carbon pricing. Among these instruments, the European Union Emissions Trading System (EU ETS) is one of the central pillars of the EU decarbonization strategy. Launched in 2005, the EU ETS is the world’s first international carbon market and is currently in its fourth trading phase. It covers greenhouse gas emissions from around 15,000 installations in the energy sector and manufacturing industry, as well as aviation within its scope; since 2024, it also covers maritime transport. Overall, the system accounts for roughly 35--40 percent of the EU’s greenhouse gas emissions. This achievement is a spectacular application of the seminal work of \cite{Montgomery72}.
	
	However, going from a one-step deterministic model showing the equivalence between a tax on pollutants and a cap-and-trade system to a real-world application facing uncertainties and economic dynamics requires some tuning. In particular over the now twenty years of lifespan of the ETS, it has known a long period of depressed prices (around 5~\euro/t) because of an unfortunate combination of an excess of allowances in circulation and the 2008 financial crisis that led to the Great Financial Recession. 
	
	For this reason and to align the needs of the firms covered by the ETS together with the objectives of emissions reduction, the EU designed and adopted in 2015 a stabilising system named the Market Stability Reserve (MSR). The MSR is a rule-based mechanism that adjusts auction volumes in response to the total number of allowances in circulation (TNAC). In the 2014 Commission proposal \cite[sec. 6.2.3]{EC14MSR}, allowances are added to the reserve whenever the TNAC exceeds \(833\) million allowances and released whenever it falls below \(400\) million. The use of the TNAC to determine the increase and decrease of allowances to be auctionned is explicitly justified as a direct measure of the supply--demand imbalance in the EU ETS, compared to more indirect indicators such as GDP, fuel prices, or weather conditions. 
	The proposal further explains that the upper and lower thresholds were determined following consultations with stakeholders and were meant to define a range in which, according to market experience, the system could operate in an orderly manner. An associated release volume of \(100\) million allowances per year was chosen to ensure predictability and gradual adjustment, and was justified as being roughly \(5\%\) of current annual EU ETS emissions, which historical experience suggested would be sufficient even in the presence of sudden and strong increases in demand. 
	
	For the concern of this paper, it is important to notice that the MSR design was undertaken without any consideration of equilibrium carbon prices, since prices have to be set by market forces and not by regulation. 
	
	Since its creation, the EU~ETS has attracted considerable attention from academics, with studies focusing on a broad range of its characteristics, including the modeling of allowance price dynamics (\cite{Carmona10,Carmona13}), the efficiency of the mechanism compared to a tax system (\cite{Kollenberg16}) and its effectiveness as a device to reduce emissions (\cite{Eslahi25}), the role and optimisation of allowance regulation in price formation (\cite{Aid23}), and the effect of the MSR itself on market equilibrium (\cite{Chaton18}).
	For a recent complete review of the efficiency of the MSR, the reader is referred to the literature survey of \cite{Borghesi23}. This short list of papers hides the more than 2,500 articles published between 2005 and 2025 having in its title the exact expression~'EU ETS' (according to Google Scholar).  However, regarding the MSR, the literature has focused on proposing amendments to the existing MSR rules based on an assessment of its performance with respect to certain objectives (see \cite{Perino16}, \cite{Perino22} and again the thorough review of the MSR by \cite{Borghesi23}).  
	
	This paper provides, to the best of our knowledge, a first attempt to model an optimal Market Stability Reserve directly. The objective is to describe how a regulator should define the lower and upper boundaries of an MSR in order to provide firms with the flow of permits needed to cover their emissions, while remaining consistent with long-run decarbonization objectives. In our model, the total number of allowances in circulation is represented by a controlled consumption process: allowances are depleted by emissions,  while the regulator can adjust their supply through cumulative injections and withdrawals.
	The regulator’s objective combines three components. The first one penalizes the aggregate volume of allowances, reflecting the will to operate a market under a minimal amount of allowances. The second penalizes fluctuations in the stock of allowances in circulation, capturing the idea that excessive volatility in allowance availability may lead to an excess in price volatility and destabilize the market. The third penalizes deviations from a target level, interpreted as the amount of allowances needed for the market to operate normally over a yearly horizon and also, the long-term target of emissions rate the regulator wishes to set.
	
	To account for unexpected macroeconomic or energy-market shocks, the emissions rate is allowed to switch between two regimes according to a continuous-time Markov chain. This specification captures the fact that emissions needs may differ substantially between normal periods and crisis periods, while the regulator observes the current regime and adjusts allowance injections accordingly. 
	The resulting problem is a singular stochastic control problem with regime switching, in which the regulator chooses when and how much to inject into  or withdraw from the market so as to stabilize the allowance stock while taking into account the long-run allowance level generated by its own intervention policy.
	
	From a mathematical viewpoint, the model is a stationary mean-field control (MFC) problem  with two-sided singular controls and regime switching. The mean-field interaction is through the stationary mean of the controlled allowance balance, which enters the regulator's running cost and is itself endogenously determined by the control. Since the control acts through cumulative injections and withdrawals of allowances, the optimal policy is expected to be of threshold type, with minimal interventions that keep the allowance stock inside a regime-dependent inaction region. The problem is closely related to the class of ergodic mean-field games and mean-field control problems with singular controls studied in \cite{cannerozzi2026cooperation}, \cite{cannerozzi2026stationary} and \cite{dianetti2026ergodic}.  To the best of our knowledge, stationary mean-field control problems with two-sided singular controls and regime switching have not been studied before.
	
Our approach is to reduce the stationary MFC problem to an auxiliary stationary mean-field game (MFG) with an ergodic criterion. In this auxiliary game, the mean-field term is frozen in a scalar parameter representing the long-run average stock of allowances, and  a Lagrange multiplier is introduced to enforce the constraint that the parameter coincides with the stationary mean generated by the controlled dynamics. For each value of the parameter, the representative control problem is an ergodic singular control problem with regime switching. Its solution is characterized by an auxiliary Dynkin game and a system of coupled variational inequalities, yielding regime-dependent reflecting barriers. We then impose the two stationary consistency conditions of the auxiliary MFG and show that they reduce to a single scalar fixed-point equation for the stationary mean. Solving this equation determines the endogenous stationary mean and the associated optimal intervention thresholds.
	
The paper delivers three main contributions. First, we establish existence of an optimal control for the stationary mean-field control problem and derive a complete characterization of the associated reflected policy through a unique scalar fixed point for the stationary mean. The fixed point uniquely determines both the endogenous stationary mean and the regime-dependent intervention thresholds.

	Second, we show that the optimal policy is characterized by four intervention thresholds, two in each regime. These thresholds determine when the regulator should inject allowances into the market and when it should withdraw them. The relative ordering of the two regime-dependent inaction regions is not known a priori and may vary with the model parameters, leading to different possible configurations of the intervention thresholds. Thus, rather than treating these configurations separately using smooth-paste conditions of each configurations and trying to numerically find the appropriate situation , we solve the variational inequality directly on a grid, with boundary conditions and a policy-iteration algorithm. This order-free numerical method also covers degenerate configurations, including the case where the two regimes share a common intervention band.
	
	Third, on the economic side, we provide a stylized calibration to the EU~ETS and use the model to identify the main determinants of the optimal intervention rule. The comparative statics show that volatility is the dominant driver of the width of the inaction region, whereas the frequency and duration of crisis regimes have only limited effects on the intervention thresholds. Intervention costs shift the corresponding boundaries in the expected direction (increased width), but their quantitative impact remains moderate over the range considered. Moreover, the regime-dependent thresholds remain highly close across all experiments. This observation is reinforced by a comparison with a reduced one-regime model based on the stationary average allowance-depletion rate: the resulting two-threshold rule closely reproduces the optimal regime-dependent thresholds, while generating differences of less than $1.8\%$ in stationary liquidity shortfall and less than $0.6\%$ in intervention volumes. Within our calibration, these results suggest a simple regulatory implication: conditioning the MSR intervention band on the macroeconomic regime provides little additional benefit. A parsimonious two-threshold rule, calibrated primarily to the volatility of allowance demand and to average market dynamics, can deliver virtually the same stabilization performance while preserving the transparency and predictability of the reserve mechanism.
	
The remainder of the paper is organized as follows. Section~\ref{sec:model} introduces the allowance-balance dynamics, formulates the stationary mean-field control problem, and states the main characterization of the optimal MSR policy. Section~\ref{sec:economic} provides the numerical analysis, including the baseline calibration, comparative statics, and the comparison between the regime-dependent policy and its one-regime approximation. Section~\ref{sec:mathematical_justification} develops the mathematical analysis underlying the main results: the Lagrangian reduction to an auxiliary stationary mean-field game, the solution of the frozen singular control problem, the stationary-consistency fixed point, and the verification of optimality for the original planner's problem. The paper concludes with a brief discussion, while the appendices collect the detailed proofs and the numerical implementation.

	\section{Model and main results}\label{sec:model}
	\subsection{Allowance balance, economic regimes, and MSR interventions}\label{ssec:model-specs}
	
	We consider a regulator who manages the quantity of allowances available to firms covered by an emissions trading system. The state variable $X_t$ represents the \emph{effective allowance balance} at time $t$, namely a reduced-form measure of the stock of allowances available for future compliance after accounting for current surrender needs. This interpretation is close to the allowance bank, or to the total number of allowances in circulation, which is the quantity used by the European Union to activate the Market Stability Reserve (MSR). Since we employ a diffusion approximation on the whole real line, negative values of $X$ should be interpreted as a scarcity position relative to the normalized zero-balance benchmark, rather than as a negative physical quantity of allowances. In the policy-relevant calibrations considered below, the optimally controlled state remains in a positive interval and can therefore be compared directly with observed MSR thresholds.
	
	The economy alternates between two regimes. Regime $1$ represents an expansionary state, in which emissions and compliance demand are relatively high, whereas regime $2$ represents a contractionary state, in which allowance depletion is slower. Formally, let $(\Omega,\cF,\bbF=(\cF_t)_{t\geq 0},\P)$ be a complete filtered probability space supporting a real-valued $\bbF$-Brownian motion $W$ and an independent two-state continuous-time Markov chain $\varepsilon=(\varepsilon_t)_{t\geq 0}$ with state space $\cE=\{1,2\}$ and generator $Q=(q_{ij})_{i,j\in\cE}$ given by
	\[
	Q=
	\begin{bmatrix}
		-\gamma_1 & \gamma_1\\
		\gamma_2 & -\gamma_2
	\end{bmatrix},
	\qquad \gamma_1,\gamma_2>0.
	\]
	The stationary distribution $\pi=(\pi_i)_{i=1,2}$ of $\varepsilon$ is
	\begin{equation}\label{mkv_chain:stationary_distribution}
		\pi_1\coloneqq\frac{\gamma_2}{\gamma_1+\gamma_2},
		\qquad
		\pi_2\coloneqq\frac{\gamma_1}{\gamma_1+\gamma_2}.
	\end{equation}
	
	In the absence of an MSR policy, the allowance balance follows
	\begin{equation}\label{eq:dynamics:uncontrolled}
		\dd X^0_t = -\delta_{\varepsilon_t}X^0_t \dd t + \sigma \dd W_t,
		\qquad X^0_0=x,
	\end{equation}
	where $\sigma>0$ and $\delta_i>0$, $i\in\cE$. The drift in \eqref{eq:dynamics:uncontrolled} captures the systematic depletion of the allowance bank through compliance and the gradual rebalancing of an allowance surplus toward the normalized zero-balance benchmark. The Brownian term provides a parsimonious approximation of unexpected shocks to net allowance demand arising, for example, from output, energy prices, weather, or abatement opportunities; see, e.g.,~\cite{Kollenberg16,taschini2019ETS}. The regime-dependent coefficient $\delta_{\varepsilon_t}$ allows this depletion to be faster in an expansionary economy and slower in a contractionary economy. Hence, we assume the following.
	
	\begin{assumption}[Allowance-depletion rates]\label{assumption:delta}
		It holds $\delta_1>\delta_2>0$.
	\end{assumption}
	
	The regulator can alter the allowance balance through an $\bbF$-adapted singular control $A=(A_t)_{t\geq 0}$. We write its Jordan decomposition as
	\[
	A_t=A^+_t-A^-_t,
	\]
	where $A^+$ and $A^-$ are non-decreasing, c\`adl\`ag, have disjoint supports, satisfy $A^\pm_{0-}=0$, and have integrable finite variation on every bounded time interval. The process $A^+_t$ records cumulative releases of allowances into circulation up to time $t$, whereas $A^-_t$ records cumulative placements of allowances into the reserve, or equivalently cumulative reductions in the quantity made available to the market. We abstract from an explicit inventory constraint for the reserve and interpret $A^+$ and $A^-$ as net adjustments to the supply of allowances available for compliance.
	
	For any $A=(A^+,A^-)$ singular control, the controlled allowance balance satisfies
	\begin{equation}\label{application:dynamics}
		\dd X^A_t = -\delta_{\varepsilon_t} X^A_t \dd t + \sigma \dd W_t + \dd A^+_t - \dd A^-_t,
		\qquad X^A_{0-}=x.
	\end{equation}
	The formulation permits both gradual and instantaneous interventions. In particular, we shall see that after a change of regime the regulator may immediately adjust the allowance balance if it lies outside the intervention band associated with the new regime. 
	
	For an initial condition $(x,i)\in\R\times\cE$, we write $\P_{x,i}$ for the corresponding probability measure and $\E_{x,i}$ for expectation under $\P_{x,i}$.
	We denote by $\cB$ the set of admissible singular controls, i.e. $\bbF$-adapted singular controls that satisfy the transversality condition
	\begin{equation}\label{eq:transversality_condition}
		\varlimsup_{T \to \infty} \frac{1}{T}\E_{x,i}[ \vert X^A_T \vert ] = 0.
	\end{equation}
	The transversality condition \eqref{eq:transversality_condition} will be used in the verification theorem (see Proposition \ref{prop:hjb_frozen}).
	
	\subsection{Stationary regulatory objective}\label{ssec:model-objective}
	
The MSR is intended to operate as a permanent rule rather than as a policy designed for a fixed terminal date. We therefore formulate the regulator's problem directly in steady state. Set $\cO\coloneqq\R\times\cE$ and denote by $\cP(\cO)$ the set of probability measures on $\cO$.
	
\begin{definition}\label{def:ergodic_strategies}
We say that a strategy $A\in\cB$ is admissible for the stationary mean-field control problem if the following conditions hold:
\begin{enumerate}[label=(\roman*)]
\item the joint law of $(X^A_t,\varepsilon_t)$ converges to a unique limiting stationary distribution $p^A_\infty\in\cP(\cO)$ satisfying
\[
\sum_{i=1}^2\int_{\R}|x|\,p^A_\infty(\dd x,i)<\infty;
\]
\item the long-run average allowance balance agrees with the mean of the limiting stationary distribution, namely
\begin{equation}\label{eq:stationary_occupation_property}
\lim_{T\to\infty}\frac{1}{T}\E_{x,i}\Bigg[\int_0^T X^A_t\,\dd t\Bigg] =\sum_{j=1}^2\int_{\R}x\,p^A_\infty(\dd x,j).
\end{equation}
\end{enumerate}
We denote the set of admissible strategies by $\cB_{{\rm mfc}}$ and define the stationary mean generated by $A$ as
\begin{equation}\label{eq:stationary_mean_control}
\theta^A\coloneqq\sum_{i=1}^2\int_{\R}x\,p^A_\infty(\dd x,i).
\end{equation}
\end{definition}
	
For $x,\theta\in\R$, set
\begin{equation}\label{eq:instant_cost}
l(x,\theta)\coloneqq\rho x+\phi(x-\theta)^2+\psi(\theta-\bar\theta)^2,
\end{equation}
where $\rho\in\R$, $\phi,\psi>0$, and $\bar\theta\in\R$ is the regulator's preferred long-run allowance balance. For $A\in\cB_{{\rm mfc}}$, the regulator's long-run average cost is
\begin{equation}\label{eq:cost:ergodic}
\bfJ^{{\rm mfc}}(A)
\coloneqq
\varlimsup_{T\to\infty}\frac{1}{T}\E_{x,i}\Big[ \int_0^T l(X^A_t,\theta^A) \dd t+K_+A^+_T+K_-A^-_T
\Big],
\end{equation}
where $K_+,K_->0$.
	
	The running cost in \eqref{eq:cost:ergodic} separates three regulatory concerns. The linear term $\rho X^A_t$ captures a directional preference over the size of the allowance balance. In particular, $\rho>0$ represents the cost of maintaining an excessive allowance surplus and the associated downward pressure on the carbon price, whereas $\rho<0$ assigns a benefit to market liquidity. The term $\phi(X^A_t-\theta^A)^2$ penalizes dispersion of the stationary allowance balance around its endogenous mean. It therefore discourages a stationary distribution that places substantial mass on either severe scarcity or excessive surplus; its effect on carbon-price volatility is indirect, through the stabilization of allowance availability. Finally, $\psi(\theta^A-\bar\theta)^2$ penalizes deviations of the long-run mean from the regulator's target operating level. The quantities $K_+$ and $K_-$ are the marginal costs of releasing allowances and placing allowances into the reserve, respectively.
	
Thus, the regulator trades off the average allowance balance, fluctuations of the allowance balance around its endogenous stationary mean, deviations from the policy target, and the frequency and magnitude of interventions. Note that no considerations on price control or social surplus are present in the criterion of the regulator, as reflected in the European Commission proposal described in the Introduction.
	
The stationary mean-field control problem is
\begin{equation}\label{eq:stationary_MFC_problem}
\bfJ^{{\rm mfc}}(A^\star)=\inf_{A\in\cB_{{\rm mfc}}}\bfJ^{{\rm mfc}}(A).
\end{equation}
A policy $A^\star$ satisfying \eqref{eq:stationary_MFC_problem} is called an optimal stationary MSR policy.
	
\subsection{Optimal stationary MSR rule}\label{sec:model-MSR}
	
We denote by $\cL_{(X^0,\varepsilon)}$ the infinitesimal generator of $(X^0,\varepsilon)$. For any $f:\cO\to\R$ with $f(\cdot,i)\in\dC^2(\R)$, $i=1,2$, it is given by
\begin{equation}\label{eq:infinitesimal_generator}
\cL_{(X^0,\varepsilon)}f(x,i) = -\delta_i x f_x(x,i)+\frac{\sigma^2}{2}f_{xx}(x,i) + q_{ij}\big(f(x,j)-f(x,i)\big), \qquad (x,i)\in\cO,\quad j\neq i.
\end{equation}
	
The main result shows that the regulator's optimization problem admits a regime-dependent optimal band policy, so that an optimal policy can be characterized in terms of four intervention thresholds rather than by optimizing directly over the full class of singular controls. The construction below uniquely determines these thresholds through a scalar stationary-consistency condition.
	
\begin{theorem}[Optimal stationary MSR rule]\label{thm:main_result}
Let Assumption~\ref{assumption:delta} hold. There exists an optimal policy $A^\star$ for the stationary mean-field control problem \eqref{eq:stationary_MFC_problem}. Moreover, there exist four thresholds $a^\star_+(i),a^\star_-(i)$, $i\in\cE$, with $a^\star_+(i)<a^\star_-(i)$, such that the pair $(X^{A^\star},A^\star)$ solves
\begin{equation}\label{eq:skorohod_reflection_optimal}
\left\{
\begin{aligned}
&\dd X^{A^\star}_t = -\delta_{\varepsilon_t}X^{A^\star}_t \dd t + \sigma \dd W_t + \dd A^{\star,+}_t - \dd A^{\star,-}_t,
				\qquad X^{A^\star}_{0-}=x,\\
&a^\star_+(\varepsilon_t)\leq X^{A^\star}_t\leq a^\star_-(\varepsilon_t), \qquad \P_{x,i}\text{-a.s. for a.e. }t\geq 0,\\
&\int_0^t\ind_{\{X^{A^\star}_{s-}>a^\star_+(\varepsilon_s)\}}\dd A^{\star,+}_s
= \int_0^t\ind_{\{X^{A^\star}_{s-}<a^\star_-(\varepsilon_s)\}} \dd A^{\star,-}_s =0,
\qquad \P_{x,i}\text{-a.s. for all }t\geq 0.
\end{aligned}
\right.
\end{equation}
Hence, in regime $i$, the regulator leaves the allowance balance unchanged while it lies in the inaction interval $[a^\star_+(i),a^\star_-(i)]$, releases allowances only at the lower boundary, and places allowances into the reserve only at the upper boundary. If the initial balance, or the balance immediately after a regime switch, lies outside the relevant interval, the regulator makes the smallest instantaneous adjustment that returns it to that interval.
\end{theorem}
	
The fixed-point construction stated below  uniquely determines the stationary mean $\theta^\star$ and the four reflecting thresholds of this optimal policy.
	
\begin{proposition}[Fixed-point characterization of the optimal policy]\label{prop:main_fixed_point_characterization}
For each $\theta\in\R$, define
\begin{equation}\label{eq:lagrange_multiplier}
\lambda(\theta)\coloneqq 2\psi(\bar\theta-\theta).
\end{equation}
Let $U^\theta:\R\times\cE\to\R$ denote the marginal value of an additional allowance in the auxiliary ergodic control problem associated with the candidate stationary mean $\theta$ and the multiplier $\lambda(\theta)$. For $i\in\cE$, define
\begin{equation}\label{eq:obstacle_operator_section2}
\cH_i^\theta(U)(x)
\coloneqq
-\delta_iU(x,i) +\cL_{(X^0,\varepsilon)}U(x,i) + l_x(x,\theta)-\lambda(\theta).
\end{equation}
The marginal value is characterized by the coupled double-obstacle system
\begin{equation}\label{eq:main_thm:free_boundary_problem}
\max\Big\{ \min\big\{\cH_i^\theta(U^\theta)(x),K_--U^\theta(x,i)\big\}, -K_+ - U^\theta(x,i)\Big\}=0,
\qquad x\in\R,\quad i\in\cE.
\end{equation}
Moreover, $U^\theta(\cdot,i)$ is non-decreasing and belongs to $\dC^1(\R)\cap\dC^2(\R\setminus\{a_+(i,\theta),a_-(i,\theta)\})$, where the finite thresholds
\begin{equation}\label{eq:thresholds_from_contact_sets}
a_+(i,\theta)
\coloneqq\sup\{x\in\R:U^\theta(x,i)=-K_+\},
\qquad
a_-(i,\theta)
\coloneqq\inf\{x\in\R:U^\theta(x,i)=K_-\}
\end{equation}
satisfy $a_+(i,\theta)<a_-(i,\theta)$. Equivalently,
\[
U^\theta(x,i)=-K_+\quad\text{for }x\leq a_+(i,\theta), \qquad
U^\theta(x,i)=K_-\quad\text{for }x\geq a_-(i,\theta),
\]
and $\cH_i^\theta(U^\theta)(x)=0$ in the continuation interval $(a_+(i,\theta),a_-(i,\theta))$.
		
Let $p^\theta_\infty$ be the stationary distribution of the process reflected at these thresholds and define the regime-wise cumulative distribution functions
\begin{equation}\label{eq:def_regime_wise_CDF}
\mu^\theta(x,i)
\coloneqq p^\theta_\infty\big(({-\infty},x]\times\{i\}\big),
\qquad x\in\R,\quad i\in\cE.
\end{equation}
The pair $\mu^\theta=(\mu^\theta(\cdot,1),\mu^\theta(\cdot,2))$ is the unique solution in the class of continuous, componentwise non-decreasing functions such that $\mu^\theta(\cdot,i)\in\dC^2((a_+(i,\theta),a_-(i,\theta)) \setminus \{ a_+(j,\theta),a_-(j,\theta) \} )$, $i\in\cE$, $j\neq i$, of
\begin{equation}\label{eq:regime_wise_CDF}
\left\{
\begin{aligned}
\mu^\theta(x,i)&=0,
&&x\leq a_+(i,\theta),\\
\frac{1}{2}\sigma^2\mu^\theta_{xx}(x,i) +\delta_i x\mu^\theta_x(x,i) +\sum_{j=1}^2q_{ji}\mu^\theta(x,j)&=0,
&&a_+(i,\theta)<x<a_-(i,\theta),\\
\mu^\theta(x,i)&=\pi_i,
&&x\geq a_-(i,\theta).
\end{aligned}
\right.
\end{equation}
Define the stationary-mean map
\begin{equation}\label{eq:fixed_point_map}
\Phi(\theta)
\coloneqq\sum_{i\in\cE}\int_{\R}x\,p^\theta_\infty(\dd x,i)
=\sum_{i=1}^2\left( \pi_i a_-(i,\theta) - \int_{a_+(i,\theta)}^{a_-(i,\theta)}\mu^\theta(x,i)\dd x \right).
\end{equation}
Then $\Phi$ admits a unique fixed point $\theta^\star$, and the optimal thresholds in Theorem~\ref{thm:main_result} are
\[
a^\star_\pm(i)=a_\pm(i,\theta^\star), \qquad i\in\cE.
\]
\end{proposition}
	
The construction in Proposition~\ref{prop:main_fixed_point_characterization} has a direct computational interpretation. For a candidate stationary mean $\theta$, one first solves the regime-switching obstacle problem \eqref{eq:main_thm:free_boundary_problem} and obtains a candidate intervention band in each regime. Then, one computes the invariant distribution generated by the corresponding reflected dynamics through \eqref{eq:regime_wise_CDF}. The map $\Phi$ returns the stationary mean induced by those bands, and the consistency condition $\theta=\Phi(\theta)$ selects the regulator's endogenous long-run operating level.
	
	The proof uses an auxiliary stationary mean-field game with the two scalar parameters $\theta$ and $\lambda$. This auxiliary game is solely a mathematical device for solving the regulator's control problem. For fixed $(\theta,\lambda)$, the associated ergodic singular control problem is solved by exploiting its connection with a Dynkin game. The consistency condition for $\lambda$ yields the explicit relation \eqref{eq:lagrange_multiplier}, so that the construction reduces to the scalar fixed-point equation in \eqref{eq:fixed_point_map}. Section~\ref{sec:mathematical_justification} provides the mathematical justification of this procedure and verifies that the resulting reflected policy solves the original stationary mean-field control problem.   
	
	
	\section{Economic analysis and policy implications}\label{sec:economic}
	
	\subsection{Baseline simulation}\label{ssec:economic-simulation}
	
	We use the following baseline calibration. The stock \(X\) is measured in Gt of permits in circulation and the running objective in billions of~\euro~(Bd\euro). We take \(\rho=1\)~\euro/per year, so that the objective remains in Bd\euro. For the uncontrolled stock dynamics, we set \(\delta=0.04\), corresponding to a characteristic adjustment horizon of about \(25\) years, and \(\sigma=0.5\) Gt/\(\sqrt{\mathrm{y}}\), chosen deliberately large in order to generate sufficient fluctuations in the state variable. We fix the long-run target at \(\bar\theta=0.5\) Gt, corresponding to the consumption per year of carbon emissions by the power sector of the ETS. Power utilities are the players with the highest frequency of trading and for this reason, their need in terms of liquidity is the dominant factor. For the quadratic penalties, we take \(\phi=100\) and \(\psi=100\), meaning and equal objective of reducing the variance of the allowances in circulation and the stabilization around the anticipated target. Finally, we set \(K_+=K_- = 10\) Bd\euro/Gt, i.e.\ about \(10\)~\euro/t, so that interventions are costly but not prohibitive, avoiding constant interventions.
	
	The values set above for the parameters of our model do not have the pretention to be finely calibrated on detailed data of the EU~ETS. They rely on preceeding studies ot the ETS regulation provided in \cite{Aid23}, \cite{Kollenberg16}, \cite{taschini2019ETS} and more recently \cite{Aid26}. They are set in such a way that orders of magnitude are consistent with ETS market global dynamics.
	
	With these parameters values, the Figure~\ref{fig:illustration} shows the value functions $U_i$, the densities $\mu_i$ and an example of a trajectory of the processes of allowances in circulation $X_t$, the regime state variable $\epsilon_t$, the interventions processes $A^\pm_t$. Furthermore, for these values, the corresponding thresholds in Gt are given by 
	\[
	a_+(1) = 0.230, \quad a_-(1) = 0.765, \quad a_+(2)= 0.230, \quad a_-(2)= 0.760.
	\]
	First we observe that  $a_+(1)$ and $a_+(2)$ are numerically close (in fact, up to a precision higher than $10^{-6}$). This proximity makes an algorithm based on fixing {\em ex ante} the order in which finding the fixed point more difficult to implement and justifies the numerical strategy of directly solving the variational inequality using a grid discretisation of space and an iteration policy algorithm. Second, we observe that the range of low and high values are of the order of magnitude of the true band of no-intervention, namely $0.4$ and $0.8$~Gt.

These results should not be interpreted as the claim that our model {\em does} reproduce the observed thresholds of the MSR, but as the fact that this reduced model {\em can} generate reasonable values for the intervention band of the MSR design.
	
\begin{figure}[!hbt]
\begin{center}
\begin{tabular}{ccc}
(a) & (b)  & (c) \\
\includegraphics[width=0.3\textwidth]{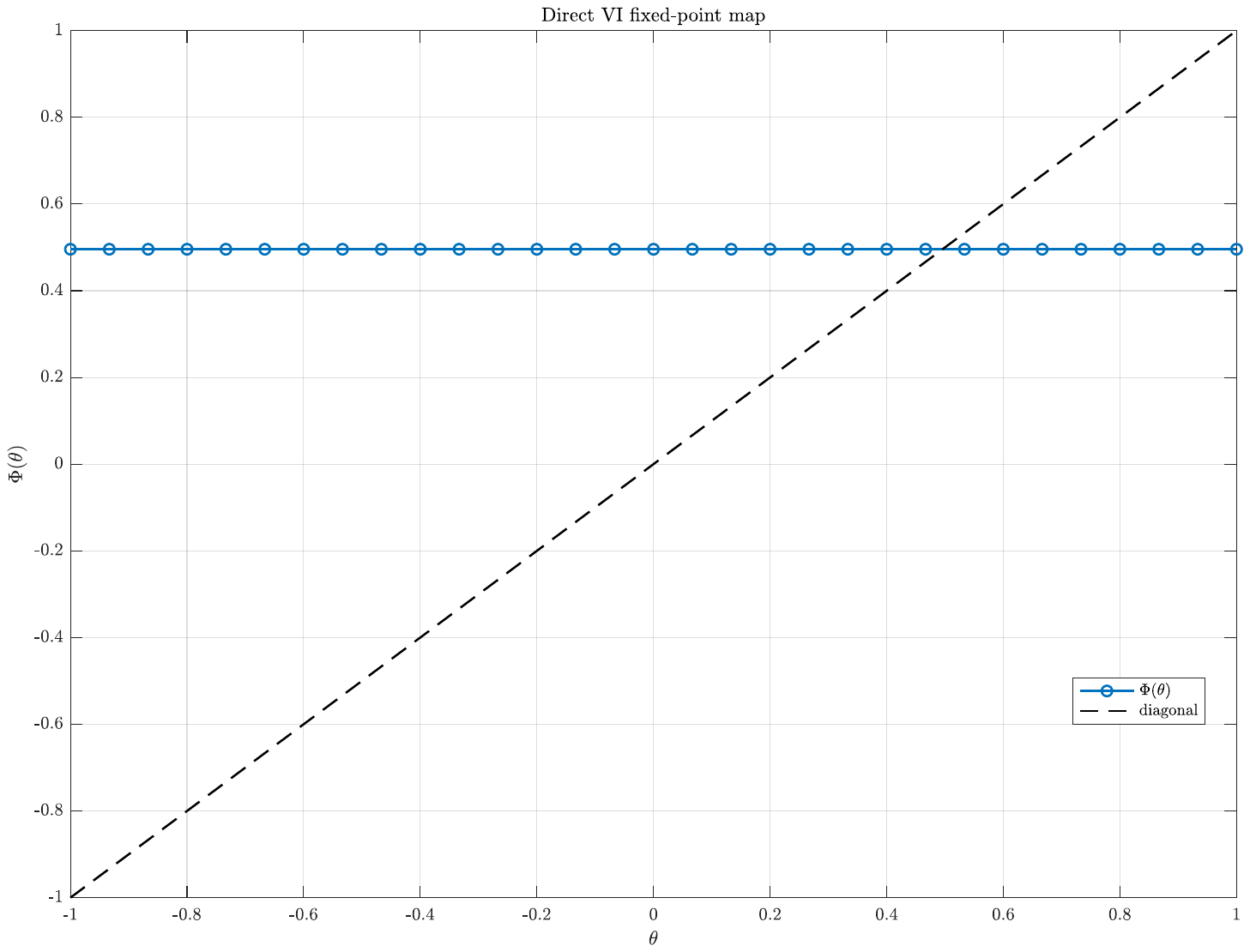} &
\includegraphics[width=0.3\textwidth]{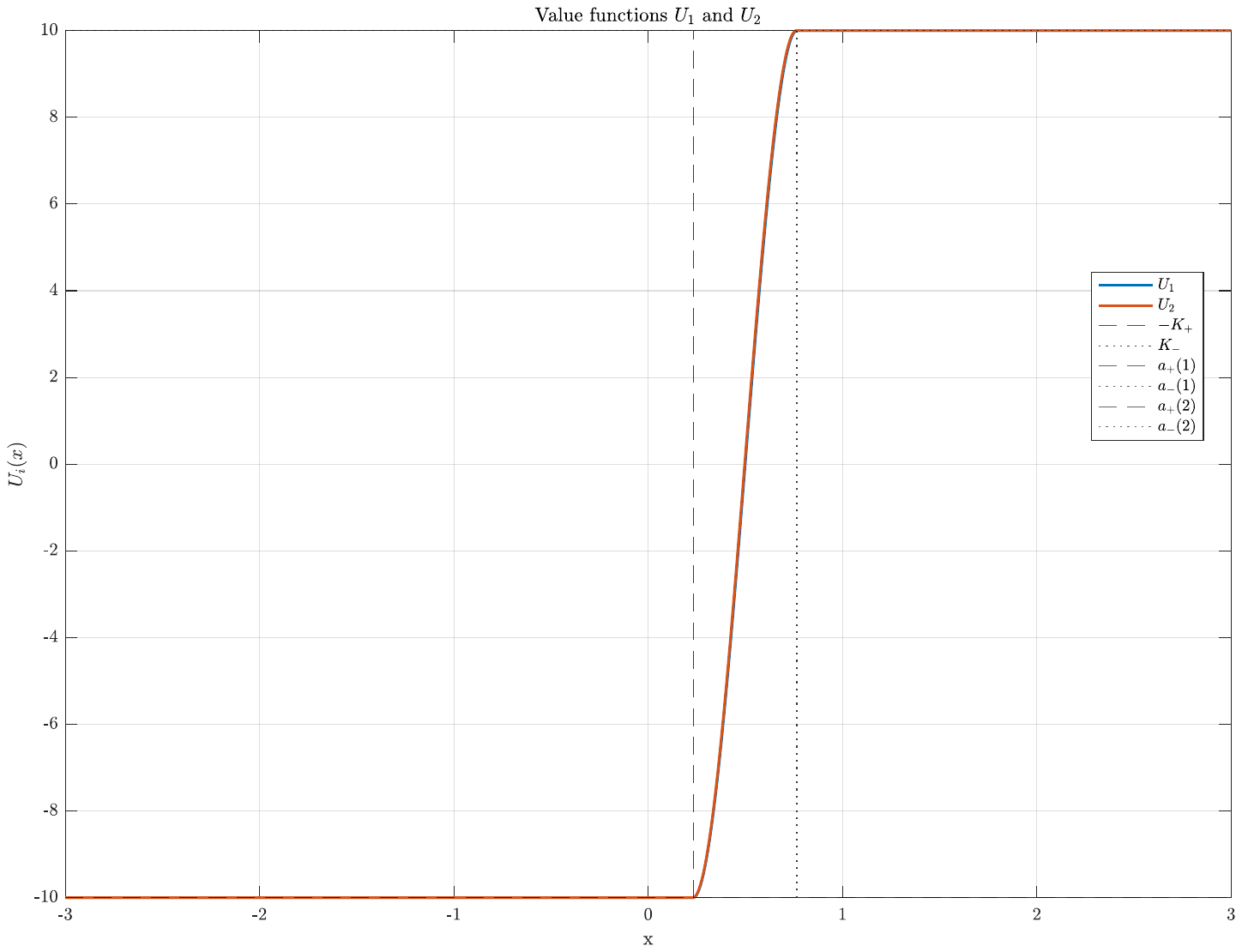}&
\includegraphics[width=0.3\textwidth]{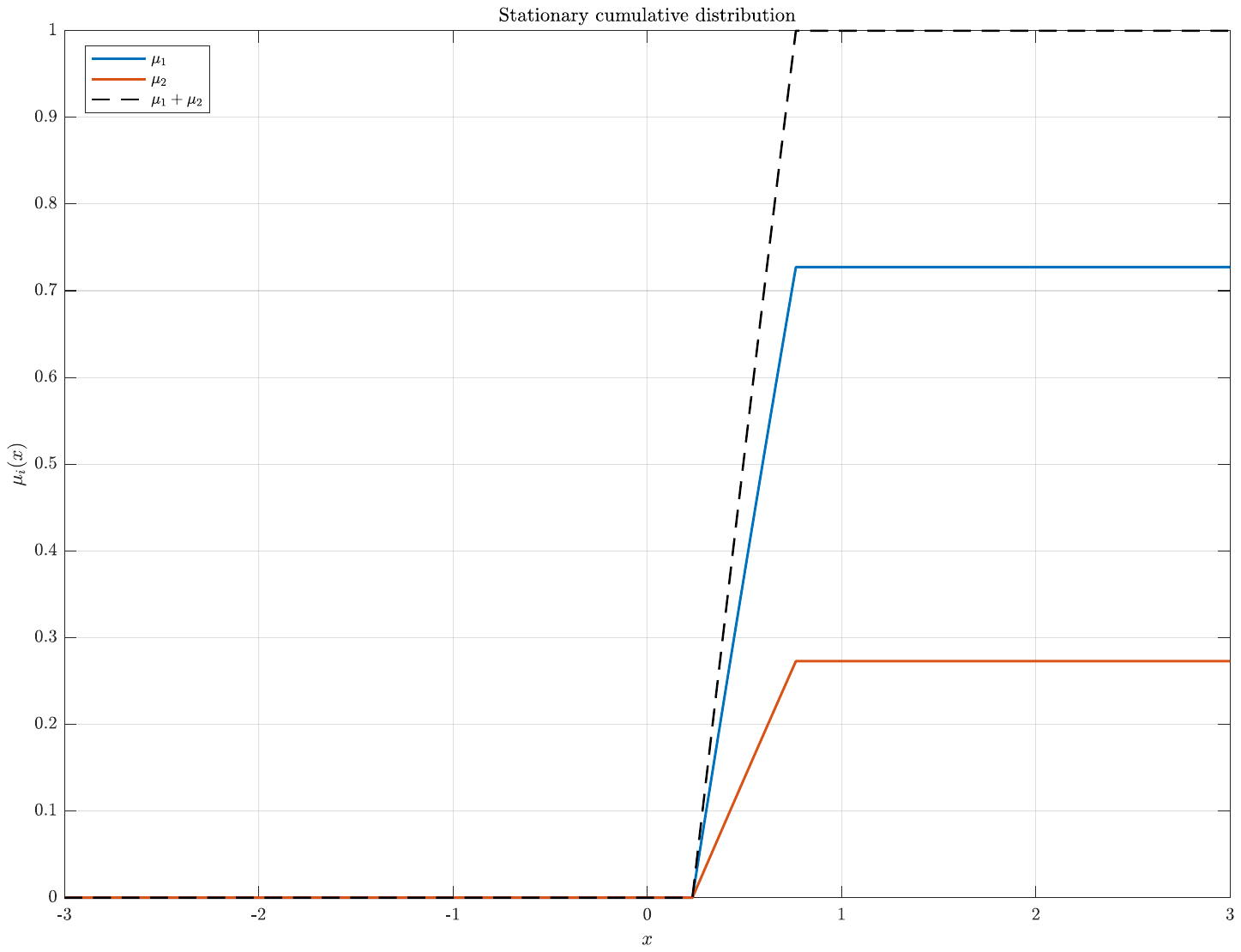}\\
(d) & (e) &   \\
\includegraphics[width=0.3\textwidth]{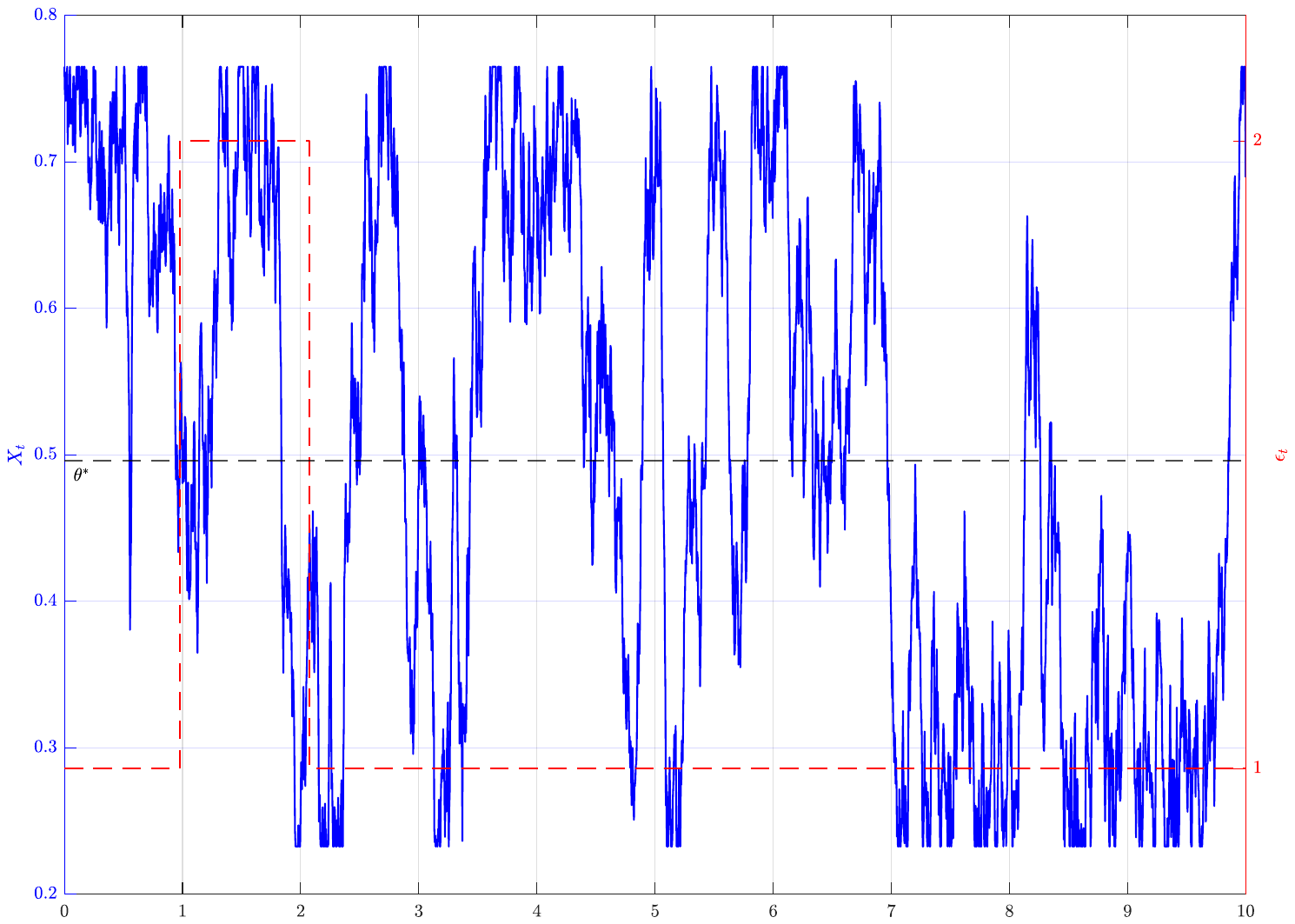} &
\includegraphics[width=0.3\textwidth]{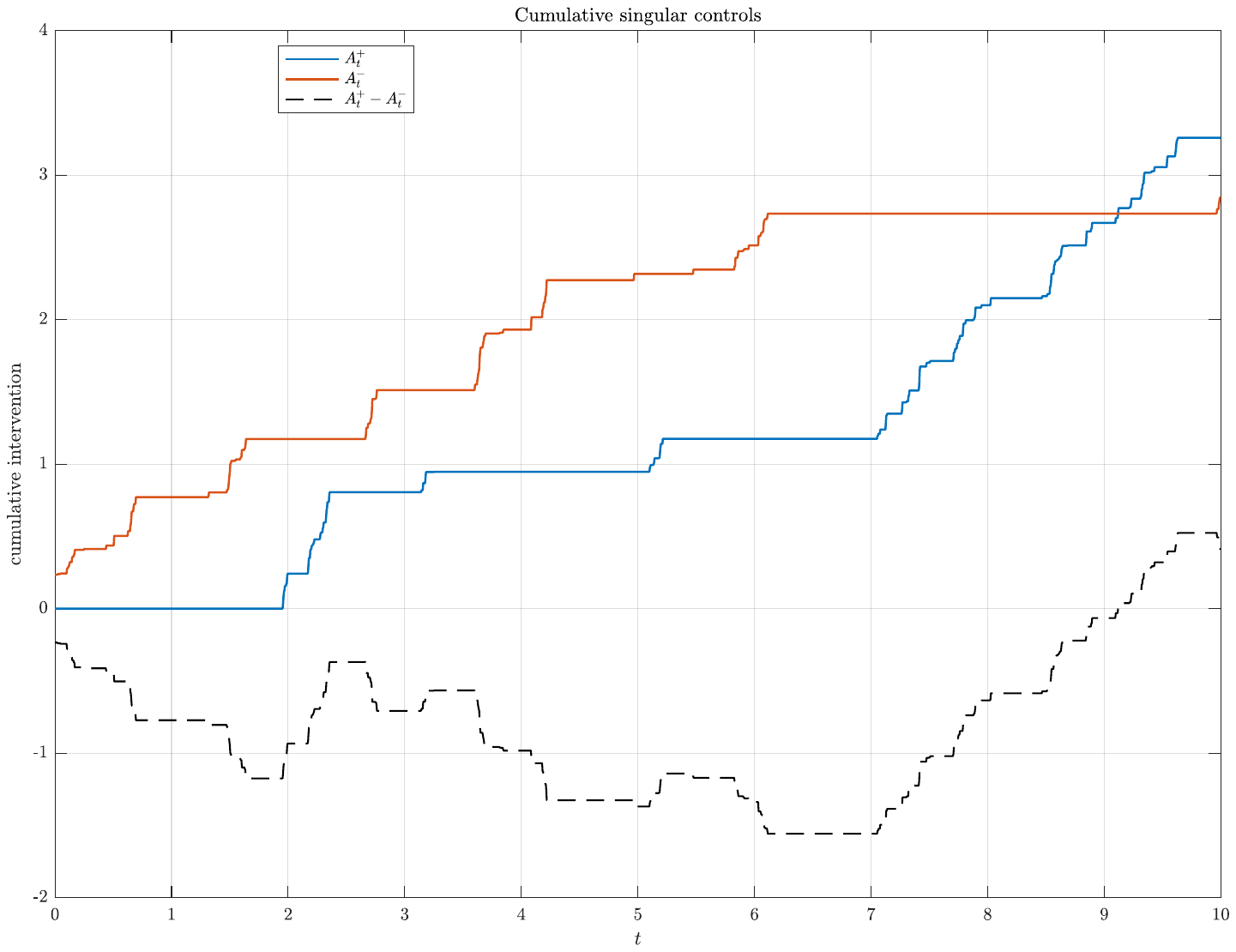}  \\
\end{tabular}
\end{center}
\caption{(a) The function $\Phi(\theta)$ (b) the value functions $U_1$ and $U_2$ (c) the functions $\mu_1$ and $\mu_2$ (d) a trajectory of allowances in circulation $X_t$ together with the regime state $\epsilon_t$ (e) the corresponding trajectory of the processes $A^+_t$ and $A^-_t$. Parameters values: $\delta_1=0.04$, $\delta_2=0.01$, $\gamma_1=1/8$, $\gamma_2=1/3$, $\sigma=0.5$, $\rho=1$ $\phi=100$, $\psi=100$, $\bar\theta=0.5$, $K_+ = K_- =10$.}
\label{fig:illustration}
\end{figure}
	
	\subsection{Comparative statics}\label{ssec:economic-comparison}
	
	Figure~\ref{fig:thresholds} reports comparative statics for the four intervention thresholds. The purpose of this exercise is not to provide a full calibration of the EU ETS, but rather to identify which parameters of the model are most relevant for the size and position of the optimal intervention band. The baseline parameters are $\delta_1=0.04$, $\delta_2=0.01$, $\gamma_1=0.50$, $\gamma_2=0.10$, $\sigma=0.5$, $\rho=1$, $\phi=\psi=100$, $\bar\theta=0.5$, and $K_+=K_-=10$. We consider four perturbations around this benchmark: the volatility $\sigma$, the frequency of crisis episodes through $\gamma_1$, their average duration through $\gamma_2^{-1}$, and the intervention costs $K_+$ and $K_-$. 
	
	The main conclusion is that the volatility of the allowance stock is the dominant parameter for the sizing of the intervention band. Higher volatility substantially widens the inaction region, while changes in the frequency or duration of non-anticipated crisis regimes have a much smaller effect on the location of the thresholds. The intervention costs have the expected directional effect, with larger costs making the regulator less willing to intervene on the corresponding side of the band, but their quantitative impact remains moderate in the range considered. Finally, across all experiments, the relative distance between regime-dependent thresholds is remarkably stable. In particular, these comparative statics mainly affect the common level or width of the intervention bands, rather than generating large differences between the thresholds associated with the two regimes. This result sustains a MSR design without considering distinct ranges of operations depending on the macro-economic regime of the economy.
	
	\begin{figure}[!hbt]
		\begin{center}
			\begin{tabular}{ccc}
				(a) & (b) & (c)  \\
				\includegraphics[width=0.3\textwidth]{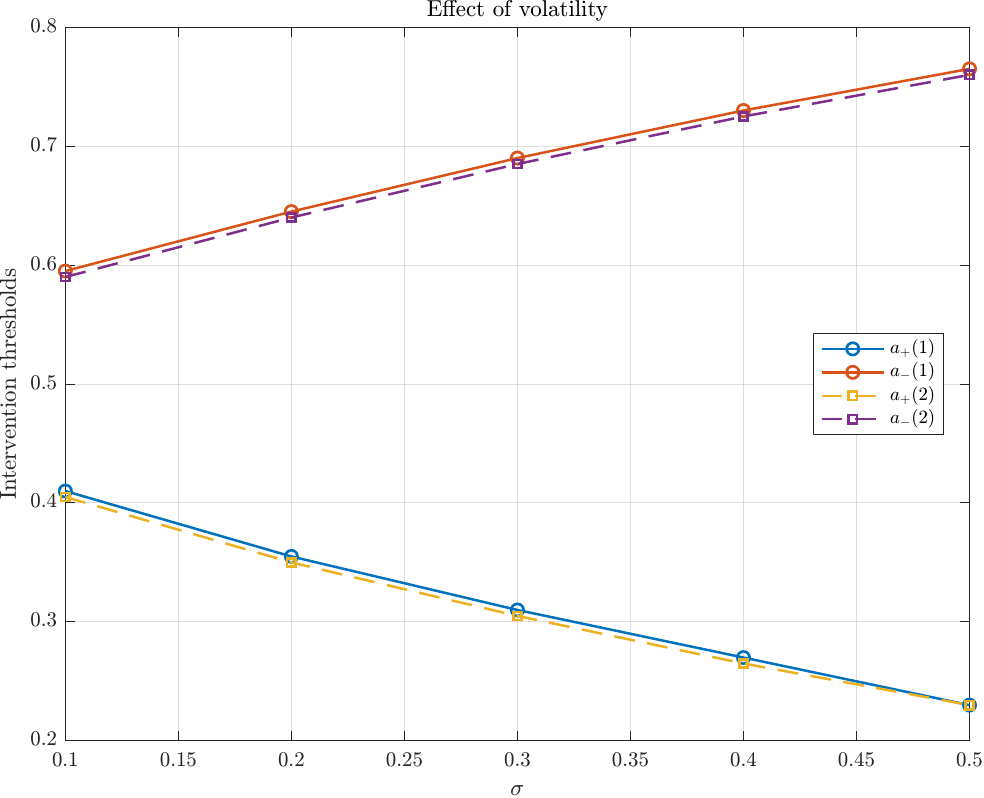} &
				\includegraphics[width=0.3\textwidth]{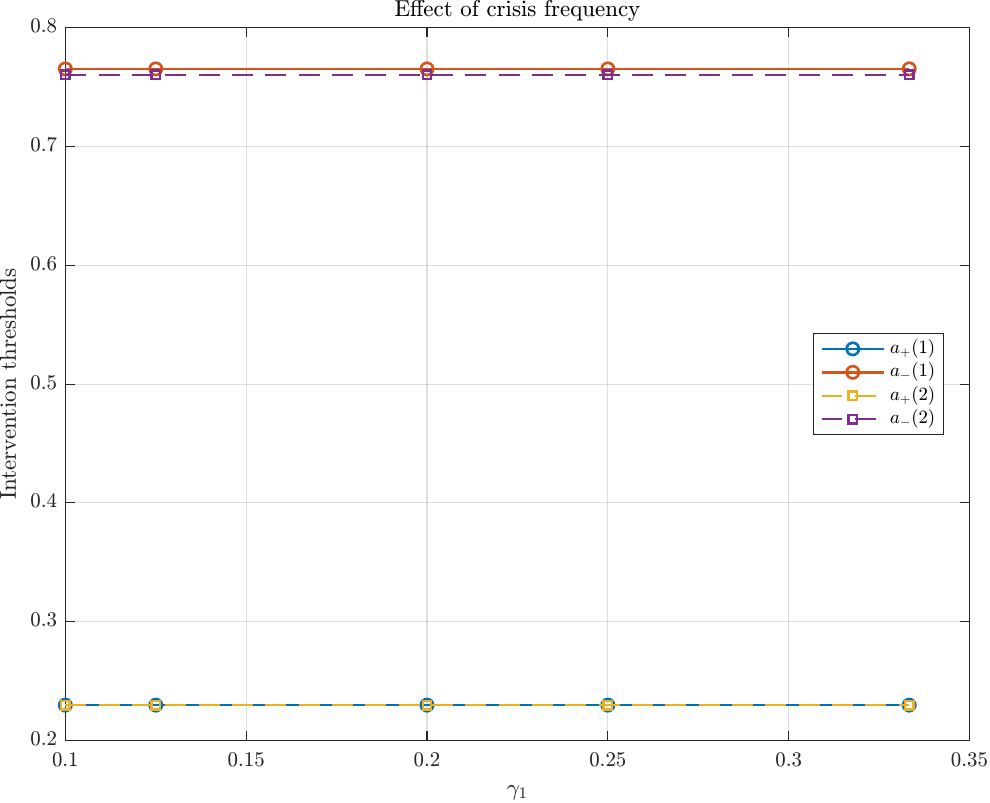}&
				\includegraphics[width=0.3\textwidth]{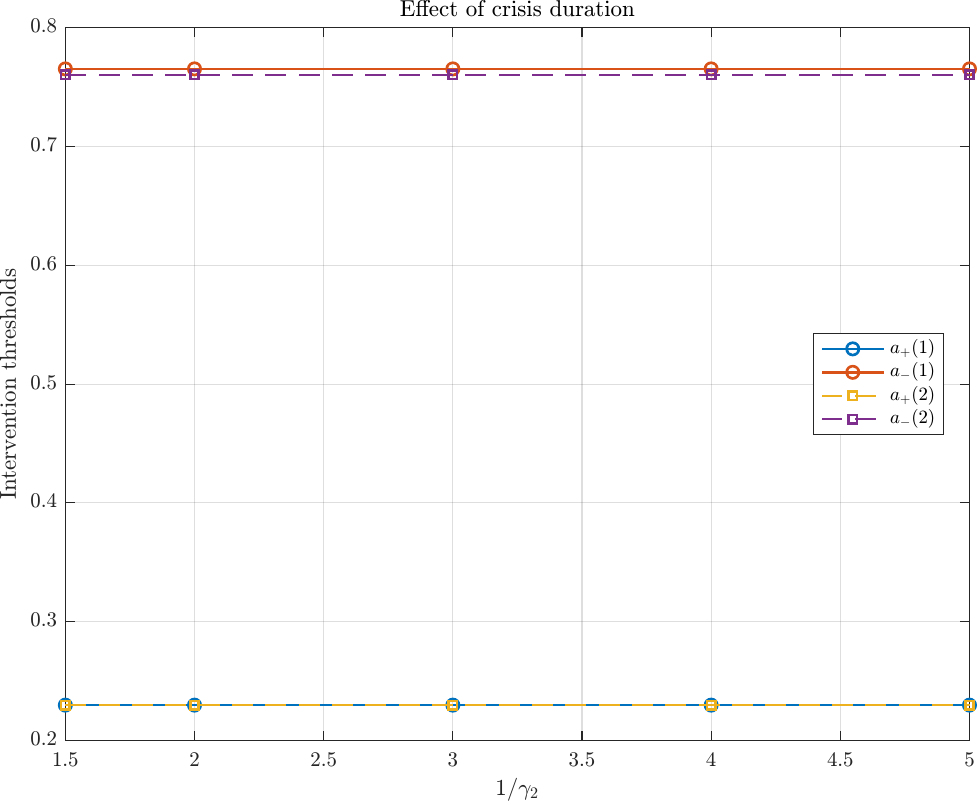}\\
				(d) & (e) &   \\
				\includegraphics[width=0.3\textwidth]{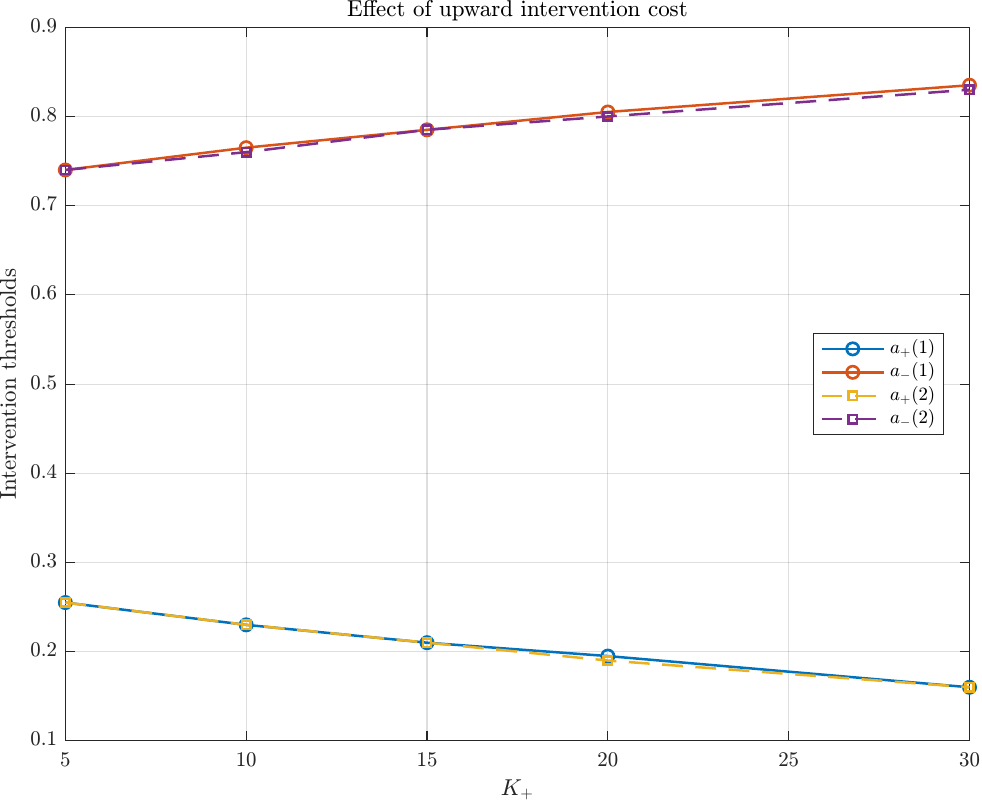}&
				\includegraphics[width=0.3\textwidth]{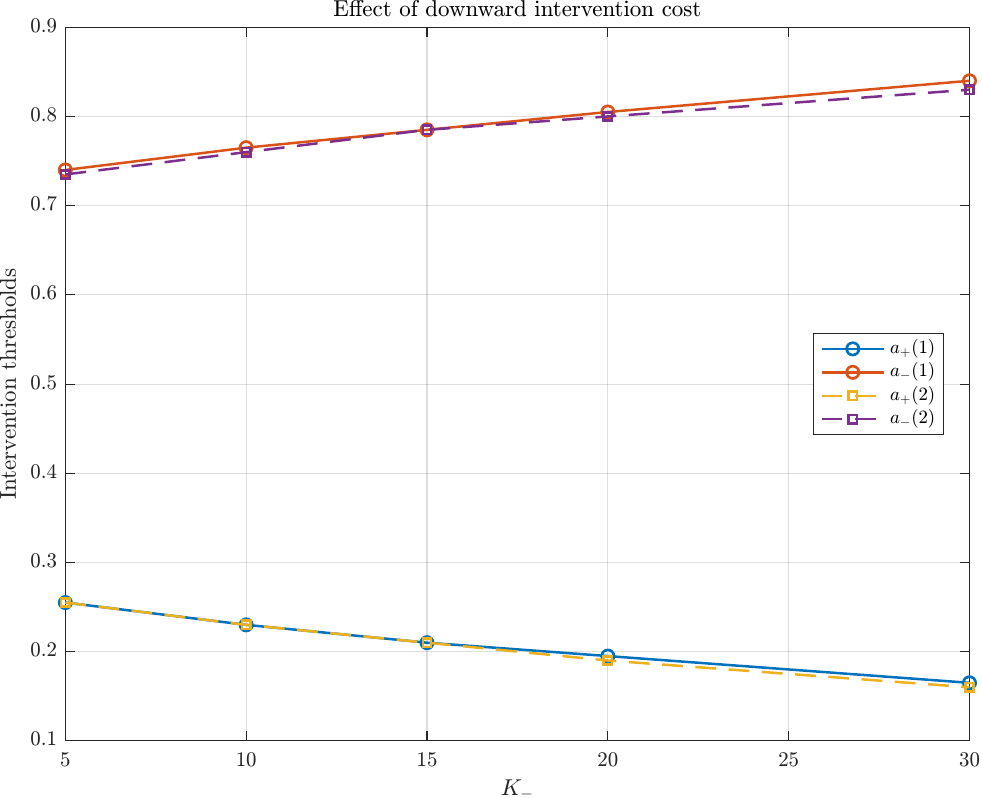} & \\
			\end{tabular}
		\end{center}
		\caption{Intervention thresholds comparative statics (a) Effect of the volatility $\sigma$ (b) Effect of the crisis frequency (c) Effect of crisis duration (d) Effect of intervention cost $K_+$ (e) Effect of intervention cost $K_-$}
		\label{fig:thresholds}
	\end{figure}
	\begin{table}[!hbt]
		\centering
		\small
		\begin{tabular}{cccccccc}
			\hline
			$\delta_2$ & $\pi_2$ & $a_+(1)$ & $a_+(2)$ & $b_+$ & $a_-(1)$ & $a_-(2)$ & $b_-$ \\
			\hline
			0.000 & 0.010 &  0.235 &  0.230 &  0.235 &  0.770 &  0.760 &  0.770 \\
			0.020 & 0.010 &  0.235 &  0.230 &  0.235 &  0.770 &  0.765 &  0.770 \\
			0.040 & 0.010 &  0.235 &  0.230 &  0.235 &  0.770 &  0.765 &  0.770 \\
			0.060 & 0.010 &  0.235 &  0.235 &  0.235 &  0.770 &  0.770 &  0.770 \\
			\hline
			0.000 & 0.094 &  0.235 &  0.230 &  0.235 &  0.770 &  0.760 &  0.765 \\
			0.020 & 0.094 &  0.235 &  0.230 &  0.235 &  0.770 &  0.765 &  0.770 \\
			0.040 & 0.094 &  0.235 &  0.230 &  0.235 &  0.770 &  0.765 &  0.770 \\
			0.060 & 0.094 &  0.235 &  0.235 &  0.235 &  0.770 &  0.770 &  0.770 \\
			\hline
			0.000 & 0.200 &  0.235 &  0.230 &  0.230 &  0.770 &  0.760 &  0.765 \\
			0.020 & 0.200 &  0.235 &  0.230 &  0.235 &  0.770 &  0.765 &  0.765 \\
			0.040 & 0.200 &  0.235 &  0.230 &  0.235 &  0.770 &  0.765 &  0.770 \\
			0.060 & 0.200 &  0.235 &  0.235 &  0.235 &  0.770 &  0.770 &  0.770 \\
			\hline
		\end{tabular}
		\caption{Intervention thresholds in the two-regime model and in the averaged one-regime model. The two-regime thresholds are denoted by $a_\pm(1)$ and $a_\pm(2)$. The averaged one-regime thresholds are denoted by $b_\pm$, with $\bar\delta=(1-\pi_2)\delta_1+\pi_2\delta_2$.}
		\label{tab:average_regime_thresholds}
	\end{table}
	
	We finally compare the two-regime model with a reduced one-regime approximation in which the regime-dependent mean-reversion rates are replaced by their stationary average. We fix $\delta_1=0.08$ and consider selected values of $\delta_2$ below $\delta_1$, namely $\delta_2 \in \{0,0.02,0.04,0.06\}.$ For the transition intensities, rather than reporting $\gamma_1$ and $\gamma_2$ directly, we use the stationary probability of regime $2$, $\pi_2= \gamma_1/(\gamma_1+\gamma_2),$ and consider the values $\pi_2\in\{0.010,0.094,0.200\}.$ For each pair $(\delta_2,\pi_2)$, the averaged one-regime model is obtained by setting
	$\bar\delta(\delta_2,\pi_2):=(1-\pi_2)\delta_1+\pi_2\delta_2.$ Table~\ref{tab:average_regime_thresholds} reports the four thresholds of the two-regime model, denoted by $a_\pm(1)$ and $a_\pm(2)$, together with the thresholds $b_\pm$ obtained in the averaged one-regime approximation.

	The comparison shows that the two-regime thresholds are very close to those obtained in the averaged one-regime model. Across the reported configurations, the lower intervention thresholds remain essentially unchanged, with $a_+(1)$ and $b_+$ equal to either $0.230$ or $0.235$~Gt, while $a_+(2)$ differs by at most one grid step. The same conclusion holds for the upper thresholds: the values remain concentrated between $0.760$ and $0.770$~Gt, and the averaged threshold $b_-$ lies within the range spanned by $a_-(1)$ and $a_-(2)$. Thus, for this calibration, non-anticipated regime-switching shocks have only a limited quantitative effect on the optimal injection and withdrawal thresholds. Their main effect is already captured by the stationary average decay rate $\bar\delta$. This suggests that the two-regime specification is useful for representing crisis dynamics and regime-dependent paths, but it does not substantially alter the optimal size of the intervention band relative to the averaged one-regime benchmark.
	
	\medskip
	
	To assess whether the proximity of the intervention thresholds also translates into similar market outcomes, we compare the two policies in terms of liquidity and intervention intensity. Both policies are evaluated under the true two-regime dynamics. The first policy is the optimal regime-dependent policy $A^*$, while the second, denoted by $A^{\mathrm{av}}$, uses the common band obtained from the averaged model with $\bar\delta=(1-\pi_2)\delta_1+\pi_2\delta_2.$ We measure liquidity shortage by $L(A):=\E[(\bar\theta-X)^+],$ that is, the stationary expected shortfall of allowances relative to the target operating level $\bar\theta$. The left panel of Table~\ref{tab:policy_comparison} shows that the two policies generate almost identical liquidity outcomes. Across all configurations, the difference between $L^{\mathrm{av}}$ and $L^*$ is below $1.8\%$, and is generally much smaller.
	
	We also compare the long-run intervention intensities
	\[
	I_\pm(A):=\lim_{T\to\infty}\frac{1}{T}\E[A_T^\pm],
	\]
	where $I_+$ measures annual releases of allowances and $I_-$ annual placements into the reserve. The right panel of Table~\ref{tab:policy_comparison} again shows very small differences: the intervention intensities generated by the averaged policy differ from those of the optimal regime-dependent policy by less than $0.6\%$ over all configurations.
	
	\begin{table}[!hbt]
		\centering
		
		\begin{minipage}[t]{0.36\textwidth}
			\vspace{0pt}
			\centering
			\textbf{(a) Liquidity shortfall}
			\smallskip
			
			\scriptsize
			\setlength{\tabcolsep}{2.5pt}
			\begin{tabular}{ccccc}
				\hline
				$\delta_2$ & $\pi_2$ & $L^*$ & $L^{\mathrm{av}}$ & $\Delta L/L^*$ (\%) \\
				\hline
				0.000 & 0.010 & 69 & 69 & -0.05 \\
				0.020 & 0.010 & 69 & 69 & -0.04 \\
				0.040 & 0.010 & 69 & 69 & -0.04 \\
				0.060 & 0.010 & 69 & 69 &  0.00 \\
				\hline
				0.000 & 0.094 & 69 & 69 & -0.43 \\
				0.020 & 0.094 & 69 & 69 & -0.34 \\
				0.040 & 0.094 & 69 & 69 & -0.34 \\
				0.060 & 0.094 & 69 & 69 &  0.00 \\
				\hline
				0.000 & 0.200 & 69 & 68 & -1.76 \\
				0.020 & 0.200 & 69 & 68 & -1.57 \\
				0.040 & 0.200 & 69 & 68 & -1.57 \\
				0.060 & 0.200 & 69 & 68 & -0.84 \\
				\hline
			\end{tabular}
		\end{minipage}
		\hfill
		\begin{minipage}[t]{0.61\textwidth}
			\vspace{0pt}
			\centering
			\textbf{(b) Intervention intensities}
			\smallskip
			
			\scriptsize
			\setlength{\tabcolsep}{2.2pt}
			\begin{tabular}{cccccccc}
				\hline
				$\delta_2$ & $\pi_2$ & $I^*_+$ & $I^{\mathrm{av}}_+$ & $\Delta I_+/I^*_+$ (\%)
				& $I^*_-$ & $I^{\mathrm{av}}_-$ & $\Delta I_-/I^*_-$ (\%) \\
				\hline
				0.000 & 0.010 & 250 & 250 &  0.06 & 212 & 211 & -0.24 \\
				0.020 & 0.010 & 250 & 250 &  0.07 & 212 & 211 & -0.23 \\
				0.040 & 0.010 & 250 & 250 &  0.07 & 211 & 211 & -0.23 \\
				0.060 & 0.010 & 250 & 250 &  0.07 & 211 & 211 & -0.23 \\
				\hline
				0.000 & 0.094 & 249 & 249 & -0.02 & 214 & 213 & -0.31 \\
				0.020 & 0.094 & 249 & 249 &  0.07 & 213 & 212 & -0.22 \\
				0.040 & 0.094 & 249 & 250 &  0.07 & 212 & 212 & -0.23 \\
				0.060 & 0.094 & 250 & 250 &  0.07 & 212 & 211 & -0.23 \\
				\hline
				0.000 & 0.200 & 247 & 247 & -0.17 & 216 & 215 & -0.51 \\
				0.020 & 0.200 & 248 & 248 &  0.02 & 215 & 214 & -0.33 \\
				0.040 & 0.200 & 249 & 249 &  0.02 & 214 & 213 & -0.36 \\
				0.060 & 0.200 & 249 & 249 &  0.01 & 212 & 212 & -0.37 \\
				\hline
			\end{tabular}
		\end{minipage}
		
		\caption{Comparison between the optimal regime-dependent policy and the averaged one-regime policy, both evaluated under the true two-regime dynamics. Left: stationary liquidity shortfall $L(A)=\E[(\bar\theta-X)^+]$, in Mt. Right: long-run intervention intensities $I_\pm(A)$, in Mt/year. Relative differences are measured with respect to the optimal regime-dependent policy.}
		\label{tab:policy_comparison}
	\end{table}
	
	Hence, for the EU ETS calibration considered here, accounting explicitly for the current economic regime has  little effect on the operation of the reserve. A two-threshold rule based on the stationary average depletion rate reproduces not only the optimal intervention band rather accuratly, but also the resulting liquidity provision and intervention volumes.
	
	
	\section{Analysis of the optimal reserve policy}\label{sec:mathematical_justification}
	
	This section provides the mathematical justification of the fixed-point construction in Proposition~\ref{prop:main_fixed_point_characterization}. The argument is based on a Lagrangian decomposition of the regulator's stationary mean-field control problem. We first freeze a candidate stationary mean and solve the resulting regime-switching singular control problem. We then impose stationary consistency and finally verify that the policy selected by the fixed point is optimal for the original planner's problem. The complete proofs are collected in Appendix~\ref{app:ergodic_singular_ctrl} for the frozen singular control problem, Appendix~\ref{app:fixed_point} for the fixed-point analysis, and Appendix~\ref{app:stationary_MFC_problem} for the verification of the stationary planner's problem.
	
	\subsection{Lagrangian reduction of the planner's problem}\label{subsec:lagrangian_reduction}
	
	The running cost in \eqref{eq:cost:ergodic} depends on the control both directly, through the state process, and indirectly, through the stationary mean generated by the control. To separate these two effects, fix two scalars $\theta,\lambda\in\R$. For $A\in\cB$, consider the auxiliary criterion
	\begin{equation}\label{eq:mfg:payoff}
		\bfJ(A,\theta,\lambda)
		\coloneqq
		\varlimsup_{T\to\infty}\frac{1}{T}\E_{x,i}\left[
		\int_0^T\Big(l(X^A_t,\theta)+\lambda(\theta-X^A_t)\Big) \dd t
		+K_+A^+_T+K_-A^-_T
		\right].
	\end{equation}
	Here, $\theta$ is a candidate value for the stationary mean of the allowance balance, while $\lambda$ is the shadow price associated with the requirement that the mean generated by the controlled dynamics coincide with $\theta$. If this consistency condition holds, the multiplier term has zero stationary average. Nevertheless, it changes the marginal incentives in the frozen control problem and thereby internalizes the dependence of the planner's objective on the endogenous stationary mean.
	
	The auxiliary problem has the structure of a stationary mean-field game. It should not be interpreted as a different economic model: it is introduced solely as a device for decentralizing and solving the regulator's optimization problem.
	
	\begin{definition}\label{def:mfg:solution}
		A triple $(A^\star,\theta^\star,\lambda^\star)\in\cB\times\R^2$ is an equilibrium of the auxiliary stationary MFG if the following conditions hold:
		\begin{enumerate}[label=\roman*)]
			\item $A^\star$ minimizes the frozen criterion, namely
			\[
			\bfJ(A^\star,\theta^\star,\lambda^\star)
			\leq \bfJ(A,\theta^\star,\lambda^\star),
			\qquad A\in\cB;
			\]
			\item the optimally controlled process $(X^{A^\star},\varepsilon)$ admits a unique limiting stationary distribution $p^{\theta^\star,\lambda^\star}_\infty\in\cP(\cO)$ satisfying
			\begin{align}
				\theta^\star
				&=\sum_{i=1}^2\int_{\R}x\,p^{\theta^\star,\lambda^\star}_\infty(\dd x,i),
				\label{mfg:cons:theta}\\
				\lambda^\star
				&=-\sum_{i=1}^2\int_{\R}l_\theta(x,\theta^\star)\,
				p^{\theta^\star,\lambda^\star}_\infty(\dd x,i).
				\label{mfg:cons:lambda}
			\end{align}
		\end{enumerate}
	\end{definition}
	
	The first consistency condition requires the frozen mean to agree with the stationary mean generated by the optimal frozen policy. The second is the first-order condition associated with the effect of the endogenous mean on the planner's running cost. Together, these conditions select the auxiliary problem whose optimal policy also solves the original stationary MFC problem.
	
	The remainder of the section follows three steps. First, we solve \eqref{eq:mfg:payoff} for fixed $(\theta,\lambda)$. Second, we show that the two consistency conditions reduce to a scalar fixed-point equation. Third, we use convexity and the multiplier condition to verify optimality for the original planner's problem.
	
	\subsection{The optimal policy for fixed parameters}
	\label{subsec:frozen_problem}
	
	Fix $(\theta,\lambda)\in\R^2$. The frozen problem is an ergodic singular control problem with regime switching. Its solution is obtained from a Dynkin game for the marginal value of the allowance balance, following the singular-control--Dynkin-game connection developed in \cite{dianetti2026ergodic}.
	
	For any $(x,i)\in\cO$, define
	\begin{equation}\label{eq:dynkin}
		\begin{aligned}
			U^{\theta,\lambda}(x,i) \coloneqq \inf_{\tau\geq0}\sup_{\eta\geq0}\E_{x,i}\Bigg[ \int_0^{\tau\wedge\eta} e^{-\int_0^t\delta_{\varepsilon_s} \dd s} \big(l_x(X^0_t,\theta)-\lambda\big) \dd t +K_-e^{-\int_0^\tau\delta_{\varepsilon_s}\dd s}\ind_{\{\tau<\eta\}} -K_+e^{-\int_0^\eta\delta_{\varepsilon_s} \dd s}\ind_{\{\eta<\tau\}} \Bigg],
		\end{aligned}
	\end{equation}
	where $\tau$ and $\eta$ range over all $\bbF$-stopping times. The factor
	$e^{-\int_0^t\delta_{\varepsilon_s}ds}$ is not an economic discount factor. It is the derivative of the uncontrolled flow with respect to its initial condition:
	\[
	\frac{\partial X^{0,x}_t}{\partial x}
	=e^{-\int_0^t\delta_{\varepsilon_s} \dd s}.
	\]
	Thus, $U^{\theta,\lambda}$ is the marginal value of an additional unit of allowance balance in the frozen problem.
	
	Introduce the stopping and continuation regions
	\begin{multline}\label{eq:stopping_region}
		\cS^{\theta,\lambda}_+
		=\{(x,i)\in\cO:U^{\theta,\lambda}(x,i)\leq-K_+\},
		\qquad
		\cS^{\theta,\lambda}_-
		=\{(x,i)\in\cO:U^{\theta,\lambda}(x,i)\geq K_-\},\\
		\cC^{\theta,\lambda}
		=\{(x,i)\in\cO:-K_+<U^{\theta,\lambda}(x,i)<K_-\}.
	\end{multline}
	At the lower obstacle $-K_+$, the marginal benefit of increasing the allowance balance equals the marginal cost of releasing an allowance. At the upper obstacle $K_-$, the marginal benefit of reducing the balance equals the marginal cost of placing an allowance into the reserve. Between the two obstacles, no intervention is optimal.
	
	\begin{proposition}[Marginal value and intervention thresholds]
		\label{prop:dynkin_game}
		The function $U^{\theta,\lambda}$ in \eqref{eq:dynkin} has the following properties.
		\begin{enumerate}
			\item\label{prop:dynkin_game:value}
			The Dynkin game admits a value and a saddle point.
			
			\item\label{prop:dynkin_game:thresholds}
			For each $i\in\cE$, there exist finite thresholds
			$a_+(i,\theta,\lambda)<a_-(i,\theta,\lambda)$ such that
			\begin{multline}\label{eq:filtered:regions_U}
				\cS^{\theta,\lambda}_+
				=\{(x,i)\in\cO:x\leq a_+(i,\theta,\lambda)\},
				\qquad
				\cS^{\theta,\lambda}_-
				=\{(x,i)\in\cO:x\geq a_-(i,\theta,\lambda)\},\\
				\cC^{\theta,\lambda}
				=\{(x,i)\in\cO:a_+(i,\theta,\lambda)<x<a_-(i,\theta,\lambda)\}.
			\end{multline}
			
			\item\label{prop:dynkin_game:free_boundary_pde}
			For each $i\in\cE$,
			$U^{\theta,\lambda}(\cdot,i)\in\dC^1(\R)\cap
			\dC^2(\R\setminus\{a_+(i,\theta,\lambda),a_-(i,\theta,\lambda)\})$ and
			\begin{equation}\label{eq:dynkin:pde}
				\left\{
				\begin{aligned}
					-\delta_iU^{\theta,\lambda}(x,i)
					+\cL_{(X^0,\varepsilon)}U^{\theta,\lambda}(x,i)
					+l_x(x,\theta)-\lambda&=0,
					&&a_+(i,\theta,\lambda)<x<a_-(i,\theta,\lambda),\\
					U^{\theta,\lambda}(x,i)&=-K_+,
					&&x\leq a_+(i,\theta,\lambda),\\
					U^{\theta,\lambda}(x,i)&=K_-,
					&&x\geq a_-(i,\theta,\lambda).
				\end{aligned}
				\right.
			\end{equation}
			Moreover,
			\begin{equation}\label{eq:dynkin:inequalities}
				\left\{
				\begin{aligned}
					-\delta_iU^{\theta,\lambda}
					+\cL_{(X^0,\varepsilon)}U^{\theta,\lambda}
					+l_x(\cdot,\theta)-\lambda&\geq0,
					&&x\geq a_+(i,\theta,\lambda),\\
					-\delta_iU^{\theta,\lambda}
					+\cL_{(X^0,\varepsilon)}U^{\theta,\lambda}
					+l_x(\cdot,\theta)-\lambda&\leq0,
					&&x\leq a_-(i,\theta,\lambda).
				\end{aligned}
				\right.
			\end{equation}
		\end{enumerate}
	\end{proposition}
	
	\begin{proof}
		The complete proof is given in Appendix~\ref{app:ergodic_singular_ctrl}. More precisely, item~\ref{prop:dynkin_game:value} is established there under the heading \emph{Proof of Proposition~\ref{prop:dynkin_game}--\ref{prop:dynkin_game:value}}, item~\ref{prop:dynkin_game:thresholds} under \emph{Proof of Proposition~\ref{prop:dynkin_game}--\ref{prop:dynkin_game:thresholds}}, and item~\ref{prop:dynkin_game:free_boundary_pde} under \emph{Proof of Proposition~\ref{prop:dynkin_game}--\ref{prop:dynkin_game:free_boundary_pde}}.
	\end{proof}
	
	For completeness, we state the verification result that connects these marginal thresholds to the frozen singular control problem. Define
	\begin{equation}
		\label{eq:VR}
		V^{\theta,\lambda}(x,i)
		\coloneqq
		\int_{a_+(i,\theta,\lambda)}^xU^{\theta,\lambda}(y,i)\dd y
		-R^{\theta,\lambda}(i),\qquad 
		R^{\theta,\lambda}(i)
		\coloneqq
		\int_{a_+(i,\theta,\lambda)}^{a_-(i,\theta,\lambda)}
		U^{\theta,\lambda}(y,i)\dd y, 
	\end{equation}
	and, for each $i\in\cE$,
	\begin{multline}
		\label{eq:kappa}
		\kappa^{\theta,\lambda}(i)
		\coloneqq
		\sum_{j\neq i}q_{ij}\left(
		\int_{a_+(j,\theta,\lambda)}^{a_+(i,\theta,\lambda)}
		U^{\theta,\lambda}(y,j)\dd y
		-R^{\theta,\lambda}(j)+R^{\theta,\lambda}(i)
		\right)\\
		+\delta_i a_+(i,\theta,\lambda)K_+
		+l(a_+(i,\theta,\lambda),\theta)
		+\lambda\big(\theta-a_+(i,\theta,\lambda)\big).
	\end{multline}
	
	\begin{proposition}[Optimal reflected policy for the frozen problem]\label{prop:hjb_frozen}
		The following statements hold.
		\begin{enumerate}
			\item\label{prop:hjb_frozen:hjb}
			The pair $(V^{\theta,\lambda},\kappa^{\theta,\lambda})$ solves
			\begin{equation}\label{eq:HJB}
				\min\left\{
				\cL_{(X^0,\varepsilon)}V^{\theta,\lambda}(x,i)
				+l(x,\theta)+\lambda(\theta-x)-\kappa^{\theta,\lambda}(i),
				-V_x^{\theta,\lambda}(x,i)+K_-,
				V_x^{\theta,\lambda}(x,i)+K_+
				\right\}=0.
			\end{equation}
			
			\item\label{prop:hjb_frozen:optimal_ctrl}
			Let $(X^{A^\star(\theta,\lambda)},A^\star(\theta,\lambda))$ solve
			\begin{equation}\label{eq:skorohod_reflection:theta_lambda}
				\left\{
				\begin{aligned}
					&\dd X^{A^\star(\theta,\lambda)}_t
					=-\delta_{\varepsilon_t}X^{A^\star(\theta,\lambda)}_t \dd t
					+\sigma \dd W_t + \dd A^{\star,+}_t(\theta,\lambda)-dA^{\star,-}_t(\theta,\lambda),
					\qquad X^{A^\star}_{0-}=x,\\
					&a_+(\varepsilon_t,\theta,\lambda)
					\leq X^{A^\star(\theta,\lambda)}_t
					\leq a_-(\varepsilon_t,\theta,\lambda),
					\qquad \P_{x,i}\text{-a.s. for a.e. }t\geq0,\\
					&\int_0^t\ind_{\{X^{A^\star(\theta,\lambda)}_{s-}>a_+(\varepsilon_s,\theta,\lambda)\}}
					\dd A^{\star,+}_s(\theta,\lambda)
					=\int_0^t\ind_{\{X^{A^\star(\theta,\lambda)}_{s-}<a_-(\varepsilon_s,\theta,\lambda)\}}
					\dd A^{\star,-}_s(\theta,\lambda)=0.
				\end{aligned}
				\right.
			\end{equation}
			Then $A^\star(\theta,\lambda)$ is optimal for the frozen problem \eqref{eq:mfg:payoff}.
			
			\item\label{prop:hjb_frozen:limsup}
			The limsup in $\bfJ(A^\star(\theta,\lambda),\theta,\lambda)$ is an ordinary limit and
			\begin{multline}\label{eq:cost:limit_exists}
				\lim_{T\to\infty}\frac{1}{T}\E_{x,i}\left[
				\int_0^T\Big(l(X^{A^\star(\theta,\lambda)}_t,\theta)
				+\lambda(\theta-X^{A^\star(\theta,\lambda)}_t)\Big)\dd t
				+K_+A^{\star,+}_T(\theta,\lambda)+K_-A^{\star,-}_T(\theta,\lambda)
				\right]\\
				=\sum_{i\in\cE}\kappa^{\theta,\lambda}(i)\pi_i.
			\end{multline}
		\end{enumerate}
	\end{proposition}
	
	\begin{proof}
		The complete proof is given in Appendix~\ref{app:ergodic_singular_ctrl}, under the heading \emph{Proof of Proposition~\ref{prop:hjb_frozen}}.
	\end{proof}
	
	Propositions~\ref{prop:dynkin_game} and \ref{prop:hjb_frozen} establish the main structural result for the frozen problem: for every pair $(\theta,\lambda)$, the optimal policy is the minimal reflection of the allowance balance inside a regime-dependent band.
	
	\subsection{Stationary consistency and the scalar fixed point}
	\label{subsec:fixed_point_auxiliary}
	
For fixed $(\theta,\lambda)$, let $p^{\theta,\lambda}_\infty$ denote the unique limiting stationary distribution of the process controlled by $A^\star(\theta,\lambda)$, which exists by \cite[Proposition 4.2]{dianetti2026ergodic}. We now impose the two consistency conditions in Definition~\ref{def:mfg:solution}.
	
Since $l_\theta(x,\theta)=2\phi(\theta-x)+2\psi(\theta-\bar\theta),$ condition \eqref{mfg:cons:lambda} yields
\[
\lambda = -2\phi\left(\theta-\sum_{i\in\cE}\int_{\R}x\,p^{\theta,\lambda}_\infty(\dd x,i)\right) - 2\psi(\theta-\bar\theta).
\]
Whenever the stationary-mean consistency condition \eqref{mfg:cons:theta} holds, this reduces to
	\begin{equation}\label{eq:lambda_theta_definition_fixed_point}
		\lambda=\lambda(\theta)\coloneqq2\psi(\bar\theta-\theta).
	\end{equation}
	Conversely, \eqref{mfg:cons:theta} together with \eqref{eq:lambda_theta_definition_fixed_point} implies \eqref{mfg:cons:lambda}. Hence, the two consistency conditions are equivalent to a single stationary-mean condition along the curve $\lambda=\lambda(\theta)$.
	
Equation \eqref{eq:lambda_theta_definition_fixed_point} also clarifies the role of the multiplier. The term $\lambda(\theta-X^A_t)$ acts on deviations of the allowance balance from the candidate stationary mean $\theta$. If $\theta<\bar\theta$, then $\lambda(\theta)>0$, so that states with $X^A_t<\theta$ increase the running cost, whereas states with $X^A_t>\theta$ reduce it. Hence, in the frozen problem, the multiplier creates an incentive to raise the allowance balance relative to the candidate mean. Conversely, if $\theta>\bar\theta$, then $\lambda(\theta)<0$, and the multiplier creates an incentive to lower the allowance balance relative to $\theta$. At stationary consistency, the long-run average of this term vanishes because the stationary mean of the controlled process is precisely $\theta$. Its role is therefore not to add a direct equilibrium cost, but to modify the marginal incentives in the frozen control problem so as to internalize the dependence of the regulator's objective on the endogenous stationary mean. Finally, the coefficient $\phi$ disappears from the equilibrium multiplier because
\[
\sum_{i\in\cE}\int_{\R}(\theta-x)\, p^{\theta,\lambda}_\infty(\dd x,i)=0
\]
under the mean-consistency condition, whereas $\phi$ continues to affect the location and width of the optimal intervention bands.
	
Set
\[
a_\pm(i,\theta)\coloneqq a_\pm(i,\theta,\lambda(\theta)), \qquad A^\theta\coloneqq A^\star(\theta,\lambda(\theta)),
\qquad p^\theta_\infty\coloneqq p^{\theta,\lambda(\theta)}_\infty.
\]
The stationary-mean map is
\begin{equation}\label{eq:Phi_definition_fixed_point_section}
\Phi(\theta) \coloneqq \sum_{i\in\cE}\int_{\R}x\,p^\theta_\infty(\dd x,i).
\end{equation}
The equilibrium condition is therefore the scalar equation
\[
\theta=\Phi(\theta).
\]
By \cite[Theorem 4.1, Claim 1]{dianetti2026ergodic}, the stationary distribution $p^\theta_\infty$, and hence $\Phi(\theta)$, can be evaluated without simulation through the regime-wise CDF system in \eqref{eq:regime_wise_CDF}. Thus, \eqref{eq:Phi_definition_fixed_point_section} agrees with the computational representation in \eqref{eq:fixed_point_map}.
	
	\begin{theorem}[Existence of a stationary-consistency fixed point]
		\label{thm:fixed_point_existence}
		The map $\Phi$ admits at least one fixed point $\theta^\star\in\R$. Consequently, the auxiliary stationary MFG admits a reflected equilibrium.
	\end{theorem}
	
	\begin{proof}
		The complete proof is given in Appendix~\ref{app:fixed_point}, under the heading \emph{Proof of Theorem~\ref{thm:fixed_point_existence}}.
	\end{proof}
	
	\begin{theorem}[Uniqueness among reflected equilibria]
		\label{thm:fixed_point_uniqueness}
		The auxiliary stationary MFG admits at most one equilibrium in the class of reflected policies. Consequently, the fixed point of $\Phi$ is unique.
	\end{theorem}
	
	\begin{proof}
		The complete proof is given in Appendix~\ref{app:fixed_point}, under the heading \emph{Proof of Theorem~\ref{thm:fixed_point_uniqueness}}.
	\end{proof}
	
	Let $\theta^\star$ denote the unique fixed point, set
	$\lambda^\star=\lambda(\theta^\star)$, and let
	$A^\star=A^{\theta^\star}$. The corresponding thresholds are
	\begin{equation}\label{eq:optimal_thresholds_auxiliary}
		a^\star_\pm(i)
		=a_\pm(i,\theta^\star,\lambda^\star)
		=a_\pm(i,\theta^\star),
		\qquad i\in\cE.
	\end{equation}
	Thus, the stationary mean and the regime-dependent reflecting bands are uniquely determined by the scalar consistency equation.
	
	\subsection{Verification for the stationary planner's problem}
	\label{subsec:mfc_verification}
	
	It remains to verify that the reflected equilibrium policy of the auxiliary problem is optimal for the original stationary MFC problem.
	
	\begin{theorem}[Verification for the stationary MFC problem]
		\label{thm:mfg_to_mfc}
		Let $(A^\star,\theta^\star,\lambda^\star)$ be the reflected equilibrium constructed above. Then $A^\star$ is optimal for the stationary mean-field control problem \eqref{eq:stationary_MFC_problem}.
	\end{theorem}
	
\begin{proof}
The complete proof is given in Appendix~\ref{app:stationary_MFC_problem}, under the heading \emph{Proof of Theorem~\ref{thm:mfg_to_mfc}}.
\end{proof}


\section{Conclusion}

In this paper, we provided a framework for the optimal design of a Market Stability Reserve in a stochastic and stationary environment. We formulate the regulator's problem as a stationary mean-field control problem with two-sided singular controls and regime switching. The optimal policy is shown to take the form of regime-dependent reflecting bands. Allowances are released when the TNAC reaches a lower threshold and withdrawn when it reaches an upper threshold. The stationary mean and the four intervention thresholds are jointly determined by a scalar consistency condition. 

The numerical analysis of our model is based on a direct resolution of the variational inequality without imposing an ex ante ordering of the intervention boundaries. The comparative statics show that the volatility of the allowance balance is the main determinant of the width of the optimal intervention band, whereas the frequency and duration of adverse economic regimes have a much smaller quantitative effect. Moreover, replacing the regime-dependent policy by a single pair of thresholds based on average economic conditions produces almost identical liquidity outcomes and intervention volumes. Hence, within our calibration, explicitly conditioning the MSR on the current macroeconomic regime brings little additional benefit. This result provides support for a parsimonious regulatory rule based on a stable intervention band, whose calibration should primarily reflect the magnitude of fluctuations in allowance demand rather than attempt to react to each change in the state of the economy.

This conclusion also complements the perspective developed in \cite{Aid23}. In that paper, allowance regulation is studied through its effect on the equilibrium carbon price and, under a quadratic specification, the optimal dynamic regulation eliminates equilibrium price volatility. The present analysis approaches the same regulatory problem from the quantity side. In the authors model, rather than controlling the equilibrium price directly, the regulator optimally stabilizes the stock of allowances available to the market through a reserve mechanism. Taken together, the two results point towards a broader view of carbon-market regulation in which the supply of allowances is used not only to ensure the long-run scarcity required by the emissions target, but also to absorb short-run shocks that would otherwise translate into excessive fluctuations in market conditions. These results should be put into perspective with the seminal work of \cite{Weitzman74} on the optimal instruments to achieve a polution reduction at the least cost in a static setting. The introduction of dynamics in polution but also in the capacity of firms to reduce their emissions enriches the problem and its potential optimal solutions depending on social objectives. Deeper research on these dynamic aspects of emissions reduction may lead to a non-binary choice to be made between instruments based on quantities (allowances) and instruments based on prices (taxes). This paper should be read as a contribution to this field of on-going research.

	
	\appendix
	\section{Appendices}
	
	\subsection{The ergodic singular control problem}\label{app:ergodic_singular_ctrl}
	\textit{Proof of Proposition \ref{prop:dynkin_game}-\ref{prop:dynkin_game:value}}.
	By a straightforward application of \cite[Theorem 2.1]{peskir2008optimal_stopping_games}, as in \cite[Theorem 3.1]{dianetti2026ergodic}, we have that
	\begin{equation*}
		U^{\theta,\lambda}(x,i) \coloneqq \sup_ {\eta \geq 0} \inf_{\tau \geq 0} \E_{x,i}\left[\int_0^{\tau \land \eta} e^{ -\int_0^t \delta_{\varepsilon_s} \, \dd s } \big( l_x(X_t,\theta) - \lambda \big) \, dt + K_- e^{- \int_0^\tau \delta_{\varepsilon_s} \, \dd s } \ind_{\tau < \eta} - K_+ e^{-\int_0^\eta \delta_{\varepsilon_s} \, \dd s } \ind_{\eta < \tau} \right]
	\end{equation*}
	so that $U^{\theta,\lambda}(x,i)$ is the value of the Dynkin game. Moreover, a saddle point for the Dynkin game exists and it is given by
	\begin{equation}\label{eq:dynkin:saddle_point}
		\tau^\star(x,i,\theta,\lambda) \coloneqq \inf\{ t \geq 0: \, (X_t,\varepsilon_t) \in \cS^{\theta,\lambda}_-\}, \quad \eta^\star(x,i,\theta,\lambda) \coloneqq \inf\{ t \geq 0: \, (X_t,\varepsilon_t) \in \cS^{\theta,\lambda}_+\},
	\end{equation}
	with $\cS^{\theta,\lambda}_\pm$ and $\cC^{\theta,\lambda}$ given by \eqref{eq:stopping_region}.
	
	\vspace{0.25cm}
	\textit{Proof of Proposition \ref{prop:dynkin_game}-\ref{prop:dynkin_game:thresholds}}.
	The proof is organized into the following steps:
	
	\textbf{Step 1.} We show that the map $x \mapsto U^{\theta,\lambda}(x,i)$ is non-decreasing for any $(i,\theta,\lambda)$ fixed.
	To this extent, denote by $(X^{x,i},\varepsilon^{i})$ the solution of \eqref{eq:dynamics:uncontrolled} starting from $(x,i)$ at time $0$, to stress the dependence on the initial data.
	We observe that $X^{x,i}$ has the following explicit representation:
	\begin{equation}\label{eq:explicit_representation}
		X^{x,i}_t = e^{-\int_0^t\delta_{\varepsilon^{i}_s} \dd s}\left( x + \sigma\int_0^t e^{\int_0^s \delta_{\varepsilon^{i}_u} du} \dd W_s \right).
	\end{equation}
	Moreover, for $\tau$, $\eta$ $\bbF$-stopping times, define
	\begin{equation}\label{eq:M_tau_eta}
		M^{\theta,\lambda}_{x,i}(\tau,\eta) \coloneqq \E_{x,i}\left[\int_0^{\tau \land \eta} e^{ -\int_0^t \delta_{\varepsilon_s} \, \dd s } \big( l_x(X_t,\theta) - \lambda \big) \, \dd t + K_- e^{- \int_0^\tau \delta_{\varepsilon_s} \, \dd s } \ind_{\tau < \eta} - K_+ e^{-\int_0^\eta \delta_{\varepsilon_s} \, \dd s } \ind_{\eta < \tau} \right],
	\end{equation}
	so that $U^{\theta,\lambda}(x,i) = \inf_{\tau \geq 0} \sup_ {\eta \geq 0} M^{\theta,\lambda}_{x,i}(\tau,\eta)$, with $U^{\theta,\lambda}$ defined by \eqref{eq:dynkin}.
	By relying on the explicit representation \eqref{eq:explicit_representation}, it follows that, for $x_1\leq x_2$, $X^{x_1,i}_t\leq X^{x_2,i}_t$ for every $t\geq 0$, $\P$-a.s.
	Since $x\mapsto l_x(x,\theta)-\lambda$ is increasing for every $(\theta,\lambda)\in\R^2$, the map $x\mapsto M^{\theta,\lambda}_{x,i}(\tau,\eta)$ is non-decreasing for every pair of stopping times $(\tau,\eta)$.
	Thus, $x\mapsto U^{\theta,\lambda}(x,i)$ is non-decreasing.
	
	\smallskip
	\textbf{Step 2.} $U^{\theta,\lambda}(x,i)$ is jointly Lipschitz continuous in $(x,\theta,\lambda)$, for any fixed $i \in \cE$.
	Let $(x,\theta,\lambda),(y,\theta',\lambda')\in\R^3$ and fix $\epsilon>0$.
	Since the Dynkin game \eqref{eq:dynkin} admits a value, it is possible to choose a stopping time $\bar\eta$ such that
	\[
	U^{\theta,\lambda}(x,i)=\sup_{\eta}\inf_{\tau}M^{\theta,\lambda}_{x,i}(\tau,\eta)\leq \inf_{\tau}M^{\theta,\lambda}_{x,i}(\tau,\bar\eta)+\epsilon\leq M^{\theta,\lambda}_{x,i}(\tau,\bar\eta)+\epsilon, \quad \forall \tau.
	\]
	Moreover, choose a stopping time $\bar\tau$ such that
	\[
	U^{\theta',\lambda'}(y,i)=\inf_{\tau}\sup_{\eta}M^{\theta',\lambda'}_{y,i}(\tau,\eta)\geq \sup_{\eta}M^{\theta',\lambda'}_{y,i}(\bar\tau,\eta)-\epsilon\geq M^{\theta',\lambda'}_{y,i}(\bar\tau,\eta)-\epsilon, \quad \forall \eta.
	\]
	Taking $\tau=\bar\tau$ in the first inequality and $\eta=\bar\eta$ in the second one, we obtain
	\[
	U^{\theta,\lambda}(x,i)-U^{\theta',\lambda'}(y,i)\leq M^{\theta,\lambda}_{x,i}(\bar\tau,\bar\eta)-M^{\theta',\lambda'}_{y,i}(\bar\tau,\bar\eta)+2\epsilon.
	\]
	Hence, since $l_x(x,\theta) = 2\phi(x-\theta) + \rho$, we get
	\[
	U^{\theta,\lambda}(x,i)-U^{\theta',\lambda'}(y,i)\leq \E\left[\int_0^\infty e^{-\int_0^t\delta_{\varepsilon^{i}_s}ds}\left(2\phi |X^{x,i}_t-X^{y,i}_t|+2\phi|\theta-\theta'|+|\lambda-\lambda'|\right)\dd t\right]+2\epsilon.
	\]
	Exchanging the roles of $(x,\theta,\lambda)$ and $(y,\theta',\lambda')$, we also obtain
	\[
	|U^{\theta,\lambda}(x,i)-U^{\theta',\lambda'}(y,i)|\leq \E\left[\int_0^\infty e^{-\int_0^t\delta_{\varepsilon^{i}_s}\dd s}\left(2\phi |X^{x,i}_t-X^{y,i}_t|+2\phi|\theta-\theta'|+|\lambda-\lambda'|\right)\dd t\right]+2\epsilon.
	\]
	Set $\underline{\delta}\coloneqq \min_{i\in\cE}\delta_i>0$.
	By \eqref{eq:explicit_representation},
	\[
	|X^{x,i}_t-X^{y,i}_t|=e^{-\int_0^t\delta_{\varepsilon^{i}_s}\dd s}|x-y|.
	\]
	Therefore,
	\[
	\E\left[\int_0^\infty e^{-\int_0^t\delta_{\varepsilon^{i}_s} \dd s}|X^{x,i}_t-X^{y,i}_t| \dd t\right]\leq |x-y|\int_0^\infty e^{-2\underline{\delta}t}\dd t=\frac{|x-y|}{2\underline{\delta}},
	\]
	and
	\[
	\E\left[\int_0^\infty e^{-\int_0^t\delta_{\varepsilon^{i}_s}\dd s}\dd t\right]\leq \int_0^\infty e^{-\underline{\delta}t}\dd t=\frac{1}{\underline{\delta}}.
	\]
	Letting $\epsilon\downarrow0$, we conclude that there exists $C>0$ such that
	\[
	|U^{\theta,\lambda}(x,i)-U^{\theta',\lambda'}(y,i)|\leq C\left(|x-y|+|\theta-\theta'|+|\lambda-\lambda'|\right),
	\]
	for every $(x,\theta,\lambda),(y,\theta',\lambda')\in\R^3$.
	This proves the joint Lipschitz continuity of $U^{\theta,\lambda}(\cdot,i)$ in $(x,\theta,\lambda)$.
	
	\smallskip
	\textbf{Step 3.} For any $(i,\theta,\lambda) \in \cE \times \R^2$ define the thresholds
	\begin{equation}\label{eq:barriers}
		a_-(i,\theta,\lambda) \coloneqq \inf\{ x \in \R: \; U^{\theta,\lambda}(x,i) \geq K_- \}, \quad a_+(i,\theta,\lambda) \coloneqq \sup\{ x \in \R: \; U^{\theta,\lambda}(x,i) \leq -K_+ \}.
	\end{equation}
	By continuity and monotonicity of $U^{\theta,\lambda}$ and exploiting the bounds $-K_+ \leq U^{\theta,\lambda}(x,i) \leq K_-$, the sets $\cS^{\theta,\lambda}_\pm$ and $\cC^{\theta,\lambda}$ defined by \eqref{eq:stopping_region} can be expressed in terms of the barriers $a_\pm(i,\theta,\lambda)$ as in \eqref{eq:filtered:regions_U}.
	
	\smallskip
	\textbf{Step 4.} We now show that $a_\pm(i,\theta,\lambda)$ are always finite and that it holds $a_+(i,\theta,\lambda) < a_-(i,\theta,\lambda)$.
	To this extent, define the functions $x_\pm(i,\theta,\lambda)$ by
	\begin{equation}\label{eq:implicit_functions}
		x_-(i,\theta,\lambda) = \frac{ 2\phi \theta + \lambda - \rho + K_-\delta_i }{2\phi}, \qquad x_+(i,\theta,\lambda) =  \frac{ 2\phi \theta + \lambda - \rho - K_+\delta_i }{2\phi},
	\end{equation}
	and notice that $x_+(i,\theta,\lambda) < x_-(i,\theta,\lambda)$ for any $(i,\theta,\lambda)$, and that
	\begin{equation}\label{eq:unique_zeros}
		l_x(x,\theta) - \lambda  - K_- \delta_i \; \left\{ \begin{aligned}
			& > 0, \quad \forall x > x_-(i,\theta,\lambda),  \\
			& = 0, \quad \forall x = x_-(i,\theta,\lambda), \\
			& < 0, \quad \forall x < x_-(i,\theta,\lambda),
		\end{aligned} \right. \qquad     l_x(x,\theta) - \lambda  + K_+ \delta_i \; \left\{ \begin{aligned}
			& > 0, \quad \forall x > x_+(i,\theta,\lambda),  \\
			& = 0, \quad \forall x = x_+(i,\theta,\lambda), \\
			& < 0, \quad \forall x < x_+(i,\theta,\lambda).
		\end{aligned} \right.
	\end{equation}
	Next, define the processes
	\begin{align}
		& \Big( e^{-\int_0^{t \land \eta^*} \delta_{\varepsilon_u} \dd u }U^{\theta,\lambda}(X_{t \land \eta^*},\varepsilon_{t \land \eta^*}) + \int_0^{t\land \eta^*}e^{-\int_0^s \delta_{\varepsilon_u} \dd u } \big( l_x(X_s,\theta) -\lambda \big) \dd s \Big)_{t \geq 0},  \label{eq:lemma_regularity:subm} \\
		& \Big( e^{-\int_0^{t \land \tau^*} \delta_{\varepsilon_u} \dd u}U^{\theta,\lambda}(X_{t \land \tau^*},\varepsilon_{t \land \tau^*}) + \int_0^{t\land \tau^*}e^{-\int_0^s \delta_{\varepsilon_u} \dd u} \big(l_x (X_s,\theta) - \lambda \big) \dd s \Big)_{t \geq 0} \label{eq:lemma_regularity:supm}
	\end{align}
	and notice that, by \cite[Theorem 2.1]{peskir2008optimal_stopping_games}, they are a sub-martingale and a super-martingale under $\P_{x,i}$ respectively.
	Let now $(x,i) \in \cS^{\theta,\lambda}_-$.
	By using the sub-martingale property of the process \eqref{eq:lemma_regularity:subm} and the bound $U^{\theta,\lambda}(x,i) \leq K_-$, we deduce
	\begin{align*}
		K_- & = U^{\theta,\lambda}(x,i) \leq \E_{x,i}\left[ e^{-\int_0^{t \land \eta^*} \delta_{\varepsilon_u} \dd u }U^{\theta,\lambda}(X_{t \land \eta^*},\varepsilon_{t \land \eta^*}) + \int_0^{t\land \eta^*}e^{-\int_0^s \delta_{\varepsilon_u} \dd u } \big( l_x(X_s,\theta) -\lambda \big) \dd s \right] \\
		& \leq \E_{x,i}\left[ K_- e^{-\int_0^{t \land \eta^*} \delta_{\varepsilon_u} \dd u } + \int_0^{t\land \eta^*}e^{-\int_0^s \delta_{\varepsilon_u} \dd u } \big( l_x(X_s,\theta) -\lambda \big) \dd s  \right] \\
		& \leq K_- + \E_{x,i}\left[  \int_0^{t\land \eta^*}e^{-\int_0^s \delta_{\varepsilon_u} \dd u } \big( l_x(X_s,\theta) -\lambda -K_- \delta_{\varepsilon_s} \big) \dd s  \right]
	\end{align*}
	for any $t \geq 0$.
	This implies
	\begin{equation*}
		0 \leq \lim_{t \to 0}\frac{1}{t}\E_{x,i}\left[ \int_0^{t}e^{-\int_0^s\delta_{\varepsilon_u} \dd u} \big( l_x(X_s,\theta) -\lambda -K_- \delta_{\varepsilon_s} \big) \ind_{s < \eta^*} \dd s \right] = l_x(x,\theta) -\lambda -K_- \delta_{i}, \;\; \forall (x,i) \in \cS^{\theta,\lambda}_-,
	\end{equation*}
	where we used dominated convergence theorem, by virtue of the representation \eqref{eq:explicit_representation}.
	Analogously, for any $(x,i) \in \cS^{\theta,\lambda}_+$, by relying on the super-martingale property of the process \eqref{eq:lemma_regularity:supm}, we deduce 
	\begin{equation}\label{eq:barriers:necessarily_separated}
		\left\{ \begin{aligned}
			& l_x(x,\theta) - \lambda -K_-\delta_i \geq 0, \quad \forall (x,i) \in \cS^{\theta,\lambda}_-,  \\
			& l_x(x,\theta) - \lambda +K_+\delta_i \leq 0, \quad \forall (x,i) \in \cS^{\theta,\lambda}_+.
		\end{aligned} \right.
	\end{equation}
	Equation \eqref{eq:barriers:necessarily_separated} yields
	\begin{equation}
		\left\{ \begin{aligned}
			& a_-(i,\theta,\lambda) \geq \inf\{ x \in \R: l_x(x,\theta) -\lambda -K_- \delta_{i} \geq 0 \} = x_-(i,\theta,\lambda), \\
			& a_+(i,\theta,\lambda) \leq \sup\{ x \in \R: l_x(x,\theta) -\lambda +K_+ \delta_{i} \leq 0 \} = x_+(i,\theta,\lambda),
		\end{aligned} \right.
	\end{equation}
	with $x_+(i,\theta,\lambda) < x_-(i,\theta,\lambda)$, which implies that $a_+(i,\theta,\lambda) < a_-(i,\theta,\lambda)$.
	To show that $a_\pm(i,\theta,\lambda)$ are finite, we show by contradiction that $\cS^{\theta,\lambda}_\pm$ are nonempty.
	Indeed, suppose that $\cS^{\theta,\lambda}_+ = \emptyset$.
	Then, for any $(x,i) \in \cO$, it holds $-K_+ < U^{\theta,\lambda}(x,i)$, which implies that $\eta^* = \infty$ $\P$-a.s.
	Therefore, 
	\begin{multline*}
		-K_+ < U^{\theta,\lambda}(x,i) = \inf_{\tau}\E_{x,i}\left[\int_0^{\tau} e^{-\int_0^t \delta_{\varepsilon_u} \dd u} \big( l_x(X_t,\theta) - \lambda \big) \dd t + K_- e^{-\int_0^\tau \delta_{\varepsilon_s} \dd s} \ind_{\{ \tau < \infty\}}\right] \\
		\leq \E_{x,i}\left[\int_0^\infty e^{-\int_0^t \delta_{\varepsilon_u} \dd u } \big( l_x(X_t,\theta) - \lambda \big) \dd t \right] \leq C_1 x + C^{\theta,\lambda},
	\end{multline*}
	for all $x \in \R$, where $C^{\theta,\lambda}$ is a positive constant depending on $(\theta,\lambda)$ and $C_1$ is a positive constant independent of $(\theta,\lambda)$.
	Hence, sending $x\to-\infty$, the right-hand side converges to $-\infty$, yielding a contradiction.
	The proof of $a_-(i,\theta,\lambda) < +\infty$ is analogous.
	
	\vspace{0.25cm}
	\textit{Proof of Proposition \ref{prop:dynkin_game}-\ref{prop:dynkin_game:free_boundary_pde}.}
	The regularity of $U^{\theta,\lambda}(\cdot,i)$ $i=1,2$, and the fact that it solves the free-boundary problem \eqref{eq:dynkin:pde} follow by adapting the proof of \cite[Proposition 3.3]{dianetti2026ergodic} to the minimization present setting.
	The sign properties \eqref{eq:dynkin:inequalities} follow from the sub and super-martingale properties of the processes in \eqref{eq:lemma_regularity:subm} and \eqref{eq:lemma_regularity:supm} respectively (see again \cite[Theorem 2.1]{peskir2008optimal_stopping_games}).
	
	\bigskip
	\textit{Proof of Proposition \ref{prop:hjb_frozen}.}
	Point \ref{prop:hjb_frozen:hjb} follows by construction, with the same computations as in  \cite[Proposition 3.4]{dianetti2026ergodic}. As for point \ref{prop:hjb_frozen:optimal_ctrl}, we notice that the Skorohod reflection problem admits a unique solution by \cite[Theorem 3.2]{dianetti2026ergodic}. Optimality follows by adapting the verification argument of \cite[Proposition 3.5]{dianetti2026ergodic} to the present minimization setting.
	We are left with proving point \ref{prop:hjb_frozen:limsup}.
	Indeed, by the verification argument leading to point \ref{prop:hjb_frozen:optimal_ctrl}, including the possible projection jumps of the reflected process, it holds, for every $T>0$,
	\begin{multline*}
		\E_{x,i}\left[\int_0^T\big(l(X^{A^\star(\theta,\lambda)}_t,\theta)+\lambda(\theta-X^{A^\star(\theta,\lambda)}_t)\big) \dd t + K_+A^{\star,+}_T(\theta,\lambda)+K_-A^{\star,-}_T(\theta,\lambda)\right] \\
		= \E_i\left[\int_0^T\kappa^{\theta,\lambda}(\varepsilon_t)\dd t\right]+V^{\theta,\lambda}(x,i)-\E_{x,i}\left[V^{\theta,\lambda}(X^{A^\star(\theta,\lambda)}_T,\varepsilon_T)\right].
	\end{multline*}
	Since $X^{A^\star(\theta,\lambda)}$ is reflected between finite boundaries, $V^{\theta,\lambda}(X^{A^\star(\theta,\lambda)}_T,\varepsilon_T)$ is uniformly bounded in $T$. Hence the last two terms, divided by $T$, vanish as $T\to\infty$. Since $\varepsilon$ is ergodic with stationary distribution $\pi$, \eqref{eq:cost:limit_exists} follows.
	
	\subsection{The fixed-point problem}\label{app:fixed_point}

	\begin{lemma}\label{lemma:theta:boundary_regularity}
		The maps $\theta \mapsto a_\pm(i,\theta,\lambda)$ are non-decreasing and bi-Lipschitz for any fixed $i \in \cE$ and $\lambda \in \R$, in the sense that
		\begin{equation}\label{eq:bilipschitz_theta}
			\vert \theta - \theta' \vert \leq \vert a_\pm(i,\theta,\lambda) - a_\pm(i,\theta',\lambda) \vert \leq  \frac{2\delta_1}{\delta_2} \vert \theta - \theta' \vert.
		\end{equation}
		Analogously, the maps $\lambda \mapsto a_\pm(i,\theta,\lambda)$ are non-decreasing and bi-Lipschitz for any fixed $i \in \cE$ and $\theta \in \R$, in the sense that
		\begin{equation}\label{eq:bilipschitz_lambda}
			\frac{1}{2\phi} \vert \lambda- \lambda' \vert \leq \vert a_\pm(i,\theta,\lambda) - a_\pm(i,\theta,\lambda') \vert \leq  \frac{2\delta_1}{\delta_2}\frac{1}{2\phi} \vert \lambda - \lambda' \vert.
		\end{equation}
	\end{lemma}
	
	\begin{proof}
		Fix $i\in\cE$ and $\lambda\in\R$, and let $\theta_1<\theta_2$. We claim that the following estimates hold:
		\begin{equation}\label{eq:double_estimate_theta}
			U(x+\theta_2-\theta_1,i,\theta_2,\lambda)\leq U(x,i,\theta_1,\lambda)\leq U\left(x+\frac{2\delta_1}{\delta_2}(\theta_2-\theta_1),i,\theta_2,\lambda\right), \qquad x\in\R.
		\end{equation}
		We first show that \eqref{eq:bilipschitz_theta} follows from \eqref{eq:double_estimate_theta}.
		We focus on $a_+(\theta,\lambda,i)$, since the reasoning for $a_-(i,\theta,\lambda)$ is completely analogous.
		Since $U(a_+(i,\theta_1,\lambda),i,\theta_1,\lambda)=-K_+$, the first inequality in \eqref{eq:double_estimate_theta}, together with the bound $U\geq-K_+$, implies $U(a_+(i,\theta_1,\lambda)+\theta_2-\theta_1,i,\theta_2,\lambda)=-K_+$ as well. Therefore, $(a_+(i,\theta_1,\lambda) + \theta_2 - \theta_1,i) \in \cS^{\theta_2,\lambda}_+$, which implies
		\[
		a_+(i,\theta_1,\lambda)+\theta_2-\theta_1\leq a_+(i,\theta_2,\lambda).
		\]
		On the other hand, if $x>a_+(i,\theta_1,\lambda)$, then $U(x,i,\theta_1,\lambda)>-K_+$.
		The second inequality in \eqref{eq:double_estimate_theta} therefore yields $U(x+\frac{2\delta_1}{\delta_2}(\theta_2-\theta_1),i,\theta_2,\lambda)>-K_+$. Thus, $(x+\frac{2\delta_1}{\delta_2}(\theta_2-\theta_1),i) \notin \cS^{\theta_2,\lambda}_+$, which implies
		\[
		x+\frac{2\delta_1}{\delta_2}(\theta_2-\theta_1)>a_+(i,\theta_2,\lambda).
		\]
		Letting $x\downarrow a_+(i,\theta_1,\lambda)$, we obtain $a_+(i,\theta_2,\lambda)\leq a_+(i,\theta_1,\lambda)+\frac{2\delta_1}{\delta_2}(\theta_2-\theta_1)$.
		Consequently,
		\begin{equation}\label{eq:theta_bound_a_plus} \theta_2-\theta_1\leq a_+(i,\theta_2,\lambda)-a_+(i,\theta_1,\lambda)\leq\frac{2\delta_1}{\delta_2}(\theta_2-\theta_1). \end{equation}
		Equation \eqref{eq:theta_bound_a_plus} implies that the map $\theta\mapsto a_+(i,\theta,\lambda)$ is non-decreasing and satisfies \eqref{eq:bilipschitz_theta}.
		It remains to prove \eqref{eq:double_estimate_theta}. By the explicit representation \eqref{eq:explicit_representation} of $X^{x,i}$, for every $r\in\R$,
		\[
		X_t^{x+r,i}-X_t^{x,i}=e^{-\int_0^t\delta_{\varepsilon_s}\,\dd s}r, \qquad t\geq0.
		\]
		Recall the definition of $M^{\theta,\lambda}_{x,i}(\tau,\eta)$ in \eqref{eq:M_tau_eta}.
		Since $l_x(x,\theta)=\rho+2\phi(x-\theta)$, for every pair of stopping times $(\tau,\eta)$ we have
		\begin{equation}\label{eq:payoff_difference_theta} 
			M^{\theta_2,\lambda}_{x+r,i}(\tau,\eta)-M^{\theta_1,\lambda}_{x,i}(\tau,\eta) =2\phi\,\E\left[\int_0^{\tau\wedge\eta}e^{-\int_0^t\delta_{\varepsilon_s}\,\dd s}\left(e^{-\int_0^t\delta_{\varepsilon_s}\,\dd s}r-(\theta_2-\theta_1)\right)\dd t\right]. 
		\end{equation}
		Taking $r=\theta_2-\theta_1$ in \eqref{eq:payoff_difference_theta}, and using $e^{-\int_0^t\delta_{\varepsilon_s} ds} \leq1$, we obtain $M^{\theta_2,\lambda}_{x+\theta_2-\theta_1,i}(\tau,\eta) - M^{\theta_1,\lambda}_{x,i}(\tau,\eta) \leq 0$, for every pair $(\tau,\eta)$, i.e. $M^{\theta_2,\lambda}_{x+\theta_2-\theta_1,i}(\tau,\eta) \leq  M^{\theta_1,\lambda}_{x,i}(\tau,\eta)$. Taking the supremum in $\eta$ and the infimum in $\tau$ yields the first inequality in \eqref{eq:double_estimate_theta}.
		To prove the second inequality in \eqref{eq:double_estimate_theta}, take $r = \frac{2\delta_1}{\delta_2}(\theta_2 - \theta_1)$ in \eqref{eq:payoff_difference_theta}, which yields
		\[
		M^{\theta_2,\lambda}_{x+\frac{2\delta_1}{\delta_2}(\theta_2 - \theta_1),i}(\tau,\eta)-M^{\theta_1,\lambda}_{x,i}(\tau,\eta) =2\phi (\theta_2 - \theta_1) \E\left[ \frac{2\delta_1}{\delta_2}\int_0^{\tau\wedge\eta}e^{-2\int_0^t\delta_{\varepsilon_s}\,\dd s} \dd t - \int_0^{\tau\wedge\eta}e^{-\int_0^t\delta_{\varepsilon_s}\,\dd s} \dd t \right]. 
		\]
		By explicit computations, exploiting $\delta_2 < 2\delta_1$ (see Assumption \ref{assumption:delta}), we have
		\begin{equation}\label{eq:discount_integral_estimate}
			\int_0^u e^{-\int_0^t\delta_{\varepsilon_s}\,\dd s} \dd t \leq \int_0^u e^{-\delta_2 t} \dd t = \frac{1-e^{-\delta_2 u}}{\delta_2} \leq \frac{1-e^{-2\delta_1 u}}{\delta_2} = \frac{2\delta_1}{\delta_2} \int_0^u e^{-2\delta_1 t} \dd t \leq   \frac{2\delta_1}{\delta_2} \int_0^u e^{ -2\int_0^t\delta_{\varepsilon_s}\,\dd s }\dd t,
		\end{equation}
		which, by evaluating in $u = \tau \land \eta$ and taking the expectation, yields
		\[
		M^{\theta_2,\lambda}_{x+\frac{2\delta_1}{\delta_2}(\theta_2 - \theta_1),i}(\tau,\eta)-M^{\theta_1,\lambda}_{x,i}(\tau,\eta) \geq 0
		\]
		This implies, by taking again the infimum in $\tau$ and supremum in $\eta$, the second estimate in \eqref{eq:double_estimate_theta}.
		
		We now turn to the dependence on $\lambda$. Fix $i\in\cE$ and $\theta\in\R$, and let $\lambda_1<\lambda_2$. We establish
		\begin{equation}\label{eq:double_estimate_lambda}
			U\left(x+\frac{\lambda_2-\lambda_1}{2\phi},i,\theta,\lambda_2\right)\leq U(x,i,\theta,\lambda_1)\leq U\left(x+\frac{2\delta_1}{\delta_2}\frac{\lambda_2-\lambda_1}{2\phi},i,\theta,\lambda_2\right), \qquad x\in\R. 
		\end{equation}
		Arguing exactly as above with the lower and upper stopping regions, \eqref{eq:double_estimate_lambda} yields
		\begin{equation}\label{eq:lambda_bound_a}
			\frac{\lambda_2-\lambda_1}{2\phi}\leq a_\pm(i,\theta,\lambda_2)-a_\pm(i,\theta,\lambda_1)\leq\frac{2\delta_1}{\delta_2}\frac{\lambda_2-\lambda_1}{2\phi}.
		\end{equation}
		Thus, the maps $\lambda\mapsto a_\pm(i,\theta,\lambda)$ are non-decreasing and satisfy \eqref{eq:bilipschitz_lambda}.
		It remains to prove \eqref{eq:double_estimate_lambda}. For every $r\in\R$ and every pair of stopping times $(\tau,\eta)$,
		\begin{equation}\label{eq:payoff_difference_lambda}
			M^{\theta,\lambda_2}_{x+r,i}(\tau,\eta)-M^{\theta,\lambda_1}_{x,i}(\tau,\eta) =\E\left[\int_0^{\tau\wedge\eta}e^{-\int_0^t\delta_{\varepsilon_s}\,\dd s}\left(2\phi e^{-\int_0^t\delta_{\varepsilon_s}\,\dd s} r - (\lambda_2-\lambda_1)\right)\dd t\right].
		\end{equation}
		Taking $r=\frac{\lambda_2-\lambda_1}{2\phi}$, and using again $e^{-\int_0^t\delta_{\varepsilon_s}\,ds}\leq1$, we obtain
		\[
		M^{\theta,\lambda_2}_{x+\frac{\lambda_2-\lambda_1}{2\phi},i}(\tau,\eta)\leq M^{\theta,\lambda_1}_{x,i}(\tau,\eta),
		\]
		and hence the first estimate in \eqref{eq:double_estimate_lambda}.
		On the other hand, taking $r=\frac{2\delta_1}{\delta_2}\frac{\lambda_2-\lambda_1}{2\phi}$, in \eqref{eq:payoff_difference_lambda} and using \eqref{eq:discount_integral_estimate}, we obtain
		\[
		M^{\theta,\lambda_1}_{x,i}(\tau,\eta)\leq M^{\theta,\lambda_2}_{x+\frac{2\delta_1}{\delta_2}\frac{\lambda_2-\lambda_1}{2\phi},i}(\tau,\eta).
		\]
		Taking the values of the Dynkin games gives the second estimate in \eqref{eq:double_estimate_lambda}.
		This concludes the proof.
	\end{proof}
	
	\begin{proof}[Proof of Theorem \ref{thm:fixed_point_existence}]
		We first prove that $\Phi$ is continuous. To this end, we would like to apply \cite[Theorem 4.1]{dianetti2026ergodic}. This result cannot be invoked verbatim, since Assumption 4.1-(3) therein requires $\sum_j q_{ji}\leq0$, a condition that is not imposed here. However, in the proof of Claim 2 this assumption is used to verify the hypothesis $(H_\Psi)$ of the maximum principle in \cite[Theorem 1]{sirakov2009}. As observed in \cite[page 3]{sirakov2009}, $(H_\Psi)$ holds whenever there exists a constant positive vector $V$ such that $\mathcal C V\leq0$, where $\mathcal C$ denotes the coupling matrix of the system. In the present setting $\mathcal C=Q^\top$. Since the Markov chain is irreducible, its invariant distribution $\pi=(\pi_i)_{i\in\cE}^\top$ satisfies $\pi_i>0$ for every $i\in\cE$ and $Q^\top\pi=0$. Hence, taking $V:=\pi/\min_{i\in\cE}\pi_i$, condition $(H_\Psi)$ is satisfied. Therefore, the maximum-principle argument in the proof of \cite[Theorem 4.1, Claim 2]{dianetti2026ergodic} remains valid in our setting.
		By Lemma \ref{lemma:theta:boundary_regularity} and the continuity of $\lambda(\cdot)$, the maps $\theta\mapsto a_\pm(i,\theta)=a_\pm(i,\theta,\lambda(\theta))$ are continuous for every $i\in\cE$. Hence, the conclusion of \cite[Theorem 4.1, Claim 2]{dianetti2026ergodic} remains valid in our setting: the map $\R\ni\theta\mapsto p_\infty^\theta\in\cP(\cO)$ is continuous. Moreover, if $\theta_n\to\theta$, the convergence of the reflecting boundaries implies that, for all sufficiently large $n$, the measures $p_\infty^{\theta_n}$ and $p_\infty^\theta$ are supported on a common compact subset of $\R\times\cE$. Hence, weak convergence implies convergence of the first moments, and therefore $\Phi(\theta_n)\to\Phi(\theta)$. Thus, $\Phi$ is continuous.
		
		We now prove that $\theta\mapsto\Phi(\theta)-\theta$ changes sign. Set
		\begin{equation}\label{eq:z_theta}
			z(\theta):=\left(1-\frac{\psi}{\phi}\right)\theta+\frac{\psi}{\phi}\bar\theta-\frac{\rho}{2\phi},\qquad \Gamma(\theta):=\lambda(\theta)\theta+\phi\theta^2+\psi(\theta-\bar\theta)^2-\phi z(\theta)^2.
		\end{equation}
		We claim that there exists a constant $C>0$, independent of $\theta$, such that
		\begin{equation}\label{eq:fixed_point:second_moment_bound}
			\sum_{i\in\cE}\int_{\R}(x-z(\theta))^2p_\infty^\theta(\dd x,i)\leq C(1+|z(\theta)|).
		\end{equation}
		Equation \eqref{eq:fixed_point:second_moment_bound} implies the conclusion. Indeed, by \eqref{eq:fixed_point:second_moment_bound}, Jensen's inequality, and exploiting the affine structure of $z(\theta)$, 
		we deduce
		\[
		|\Phi(\theta)-z(\theta)| \leq \left( \sum_{i \in \cE} \int_{\R} (x - z(\theta))^2 p^\theta_\infty(\dd x,i) \right)^\frac{1}{2} \leq \sqrt{C (1+|z(\theta)|) }  \leq  C\sqrt{1+|\theta|},\qquad \theta\in\R.
		\]
		By exploiting the definition of $z(\theta)$ in \eqref{eq:z_theta}, it follows that, for $\theta \geq 0$, we have
		\[
		\Phi(\theta)-\theta = \big( \Phi(\theta) - z(\theta) \big) + \big( z(\theta) - \theta\big) \leq-\frac{\psi}{\phi}\theta+\frac{\psi}{\phi}\bar\theta-\frac{\rho}{2\phi}+C\sqrt{1+\theta}.
		\]
		Since $\psi>0$, the right-hand side is strictly negative for all $\theta>0$ large enough. Hence there exists $\theta_+>0$ such that $\Phi(\theta_+)-\theta_+<0$.
		Similarly, for $\theta\leq0$,
		\[
		\Phi(\theta)-\theta\geq-\frac{\psi}{\phi}\theta+\frac{\psi}{\phi}\bar\theta-\frac{\rho}{2\phi}-C\sqrt{1+|\theta|}.
		\]
		Since $\psi>0$, the right-hand side is strictly positive for all $\theta<0$ with $|\theta|$ large enough. Hence there exists $\theta_-<0$ such that $\Phi(\theta_-)-\theta_->0$.
		Since $\Phi$ is continuous, the intermediate value theorem yields $\theta^\star\in(\theta_-,\theta_+)$ such that $\Phi(\theta^\star)=\theta^\star$.
		
		It remains to prove \eqref{eq:fixed_point:second_moment_bound}.
		We start by noticing that, along the curve $\lambda=\lambda(\theta)$, the running cost for the ergodic singular control problem \eqref{eq:mfg:payoff} can be written as $l(x,\theta)+\lambda(\theta)(\theta-x)=\phi(x-z(\theta))^2+\Gamma(\theta)$.
		Since $\Gamma(\theta)$ is independent of $x$ and of the control, for fixed $\theta$ the optimal control $A^\theta$ also minimizes
		\[
		\widetilde J^\theta(A):=\varlimsup_{T\to\infty}\frac{1}{T}\E_{x,i}\left[\int_0^T\phi(X_t^A-z(\theta))^2 \dd t+K_+A_T^++K_-A_T^-\right].
		\]
		Fix $\theta\in\R$ and let $(\bar X^\theta,\bar A^\theta)$ be the unique solution of the Skorohod reflection problem
		\[
		\begin{cases}
			\dd \bar X^\theta_t = -\delta_{\varepsilon_t}\bar X^\theta_t \dd t + \sigma \dd W_t + \dd\bar A^{\theta,+}_t - \dd\bar A^{\theta,-}_t,\qquad \bar X^\theta_{0-}=x,\\
			z(\theta) - 1 \leq\bar X^\theta_t\leq z(\theta) + 1,\qquad \P_{x,i}\text{-a.s. for every }t\geq0,\\
			\int_{[0,t]}\ind_{\{\bar X^\theta_{s-}> z(\theta) - 1 \}}\dd\bar A^{\theta,+}_s = \int_{[0,t]}\ind_{\{\bar X^\theta_{s-}<z(\theta) + 1 \}} \dd\bar A^{\theta,-}_s = 0,\qquad \P_{x,i}\text{-a.s. for every }t\geq0.
		\end{cases}
		\]
		We apply It\^o's formula to $(\bar X^\theta_t-z(\theta))^2$ on $[0,T]$ and obtain
		\begin{multline*}
			(\bar X^\theta_T-z(\theta))^2-(\bar X^\theta_0-z(\theta))^2 =\int_0^T\left(-2\delta_{\varepsilon_t}\bar X^\theta_t(\bar X^\theta_t-z(\theta))+\sigma^2\right)\dd t \\
			+2\sigma\int_0^T(\bar X^\theta_t-z(\theta)) \dd W_t+2\int_{(0,T]}(\bar X^\theta_t-z(\theta)) \dd\bar A_t^{\theta,+} - 2\int_{(0,T]}(\bar X^\theta_t-z(\theta)) \dd\bar A_t^{\theta,-}.
		\end{multline*}
		By the Skorohod conditions, $d\bar A^{\theta,+}$ is supported on $\{\bar X^\theta=z(\theta) - 1\}$ and $d\bar A^{\theta,-}$ on $\{\bar X^\theta=z(\theta) + 1\}$, which implies
		\[
		\int_{(0,T]}(\bar X^\theta_t-z(\theta)) \dd\bar A_t^{\theta,+} = \int_{(0,T]}(\bar X^\theta_t-z(\theta)) \ind_{ \{ \bar X^\theta_t = z(\theta) - 1 \} } \dd\bar A_t^{\theta,+} = - \int_{(0,T]} \dd\bar A^{\theta,+} = -(\bar A_T^{\theta,+} -\bar A_0^{\theta,+}),
		\]
		and analogously
		\[
		\int_{(0,T]}(\bar X^\theta_t-z(\theta)) \dd\bar A_t^{\theta,-}=\bar A_T^{\theta,-}-\bar A_0^{\theta,-}.
		\]
		Therefore,
		\begin{multline*}
			(\bar X^\theta_T-z(\theta))^2-(\bar X^\theta_0-z(\theta))^2 \\
			= \int_0^T\left(-2\delta_{\varepsilon_t}\bar X^\theta_t(\bar X^\theta_t-z(\theta))+\sigma^2\right) \dd t + 2\sigma\int_0^T(\bar X^\theta_t-z(\theta)) \dd W_t - 2(\bar A_T^{\theta,+} - \bar A_0^{\theta,+}) - 2(\bar A_T^{\theta,-} - \bar A_0^{\theta,-}).
		\end{multline*}
		Taking expectations, using that the stochastic integral has zero expectation, $(\bar X^\theta_T-z(\theta))^2\geq0$, and $\vert \bar X^\theta_0-z(\theta) \vert^2\leq1$, we obtain
		\begin{multline*}
			\E_{x,i}\left[\bar A_T^{\theta,+} + \bar A_T^{\theta,-} \right] \leq \frac{1}{2} + \bar A_0^{\theta,+} +\bar A_0^{\theta,-} + \E_{x,i}\left[\int_0^T\left(\delta_{\varepsilon_t}|\bar X^\theta_t||\bar X^\theta_t-z(\theta)| +\frac{\sigma^2}{2}\right) \dd t\right]\\
			\leq (\bar A_0^{\theta,+}+\bar A_0^{\theta,-})+\frac{1}{2}+T\left(\delta_1(|z(\theta)|+1)+\frac{\sigma^2}{2}\right).
		\end{multline*}
		For fixed $x$ and $\theta$, the quantity $\bar A_0^{\theta,+}+\bar A_0^{\theta,-}$ is finite and independent of $T$. Hence, dividing by $T$ and letting $T\to\infty$, we conclude that
		\[
		\varlimsup_{T\to\infty}\frac{1}{T}\E_{x,i}\left[\bar A_T^{\theta,+}+\bar A_T^{\theta,-}\right]\leq\delta_1(|z(\theta)|+1)+\frac{\sigma^2}{2}\leq C(1+|z(\theta)|),
		\]
		for some constant $C>0$ independent of $\theta$.
		Since $|\bar X^\theta_t-z(\theta)|\leq1$ for every $t\geq0$, we also have
		\[
		\varlimsup_{T\to\infty}\frac{1}{T}\E_{x,i}\left[\int_0^T\phi(\bar X^\theta_t-z(\theta))^2 \dd t\right]\leq\phi.
		\]
		Consequently, for a possibly different constant $C>0$, independent of $\theta$, it holds
		\[
		\widetilde J^{\theta}(A^\theta) \leq  \widetilde J^{\theta}(\bar A^\theta) = \varlimsup_{T\to\infty}\frac{1}{T}\E_{x,i}\left[\int_0^T\phi(\bar X^\theta_t-z(\theta))^2 \dd t+K_+\bar A^{\theta,+}_T + K_-\bar A^{\theta,-}_T\right]\leq C(1+|z(\theta)|),
		\]
		where the first inequality holds by the optimality of $A^\theta$ for the reduced criterion. Since the intervention costs are nonnegative, possibly redefining $C$, we obtain
		\[
		\sum_{i \in \cE} \int_{\R} (x-z(\theta))^2 p_\infty^\theta(\dd x,i) = \varlimsup_{T\to\infty}\frac{1}{T}\E_{x,i}\left[\int_0^T(X^{A^\theta}_t-z(\theta))^2 \dd t\right]\leq C(1+|z(\theta)|),
		\]
		where the first equality follows from the ergodicity of the optimally reflected process and by the definition of its invariant distribution $p_\infty^\theta$.
		This concludes the proof.
	\end{proof}

	\begin{proof}[Proof of Theorem \ref{thm:fixed_point_uniqueness}]
		Let $(A^1,\theta_1,\lambda_1)$ and $(A^2,\theta_2,\lambda_2)$ be two stationary MFG solutions in the class of reflected controls.
		For $k=1,2$, write $X^k\coloneqq X^{A^k}$ and set $h_k(x)\coloneqq l(x,\theta_k)+\lambda_k(\theta_k-x)$.
		By Proposition \ref{prop:hjb_frozen}-\ref{prop:hjb_frozen:limsup} and equation \eqref{eq:cost:limit_exists}, for each $j=1,2$ we have
		\[
		\bfJ(A^j,\theta_j,\lambda_j) = \lim_{T\to\infty}\frac{1}{T}\E\left[\int_0^T h_j(X^j_t) \dd t+K_+A^{j,+}_T+K_-A^{j,-}_T\right].
		\]
		Moreover, for every $j,k\in\{1,2\}$, the difference $h_k-h_j$ is affine. Therefore, since $(X^j,\varepsilon)$ is ergodic and its stationary mean is $\theta_j$,
		\[
		\lim_{T\to\infty}\frac{1}{T}\E\left[\int_0^T\big(h_k(X^j_t)-h_j(X^j_t)\big)\dd t\right]
		\]
		exists. Consequently, for every $j,k\in\{1,2\}$, we have
		\begin{multline*}
			\lim_{T\to\infty}\frac{1}{T}\E\left[\int_0^T\big(h_k(X^j_t)-h_j(X^j_t)\big) \dd t\right]  + \lim_{T\to\infty}\frac{1}{T}\E\left[\int_0^T h_j(X^j_t) \dd t+K_+A^{j,+}_T+K_-A^{j,-}_T\right] \\
			= \lim_{T\to\infty}\frac{1}{T}\E\left[\int_0^T h_k(X^j_t) \dd t+K_+A^{j,+}_T+K_-A^{j,-}_T\right] = \bfJ(A^j,\theta_k,\lambda_k),
		\end{multline*}
		i.e. the upper limits in the definition of $\bfJ(A^j,\theta_k,\lambda_k)$, $j \neq k$, are ordinary limits.
		Moreover, define
		\[
		Q \coloneqq \varliminf_{T\to\infty}\frac{1}{T}\E\left[\int_0^T|X^1_t-X^2_t|^2 \dd t\right].
		\]
		Fix $\epsilon\in(0,1)$ and define $A^{\epsilon}\coloneqq(1-\epsilon)A^{1}+\epsilon A^{2}$.
		Since $\cB$ is stable under convex combinations, $A^{\epsilon}\in\cB$.
		By linearity of the state dynamics, $X^{A^{\epsilon}}_t=(1-\epsilon)X^1_t+\epsilon X^2_t$, $t\geq0$.
		Since $x\mapsto h_1(x)$ is quadratic with quadratic coefficient $\phi$, it holds
		\[
		h_1(X^{A^{\epsilon}}_t)=(1-\epsilon)h_1(X^1_t)+\epsilon h_1(X^2_t)-\phi\epsilon(1-\epsilon)|X^1_t-X^2_t|^2.
		\]
		By noticing that, by the minimality of the Jordan decomposition, $A^{\epsilon,\pm}_t \leq (1-\epsilon)A^{1,\pm}_t + \epsilon A^{2,\pm}_t$ for any $t \geq 0$ $\P$-a.s., taking the upper limit as $T\to\infty$, we obtain
		\[
		\bfJ(A^{\epsilon},\theta_1,\lambda_1)\leq(1-\epsilon)\bfJ(A^1,\theta_1,\lambda_1)+\epsilon\bfJ(A^2,\theta_1,\lambda_1) - \phi\epsilon(1-\epsilon)Q.
		\]
		Since $A^1$ is optimal for the frozen problem with parameters $(\theta_1,\lambda_1)$, $\bfJ(A^1,\theta_1,\lambda_1) \leq \bfJ(A^{\epsilon},\theta_1,\lambda_1)$.
		Consequently, $\bfJ(A^2,\theta_1,\lambda_1)-\bfJ(A^1,\theta_1,\lambda_1)\geq\phi(1-\epsilon)Q$.
		Letting $\epsilon\downarrow0$, we obtain $\bfJ(A^2,\theta_1,\lambda_1)-\bfJ(A^1,\theta_1,\lambda_1)\geq\phi Q$.
		Interchanging the roles of the two equilibria, for $\epsilon\in(0,1)$ define $\widetilde A^{\epsilon}\coloneqq(1-\epsilon)A^{2}+\epsilon A^{1}$.
		Repeating the previous argument with the frozen parameters $(\theta_2,\lambda_2)$ and letting $\epsilon\downarrow0$, we obtain $\bfJ(A^1,\theta_2,\lambda_2)-\bfJ(A^2,\theta_2,\lambda_2)\geq\phi Q$.
		Adding the two inequalities gives
		\[
		\bfJ(A^2,\theta_1,\lambda_1)-\bfJ(A^1,\theta_1,\lambda_1)+\bfJ(A^1,\theta_2,\lambda_2)-\bfJ(A^2,\theta_2,\lambda_2)\geq2\phi Q.
		\]
		Since all the upper limits in the definition of $\bfJ(A^k,\theta_j,\lambda_j)$ are limits, the intervention terms cancel in the left-hand side. Moreover, notice that the map $x\mapsto h_1(x)-h_2(x)$ is affine with slope 
		\[
		-2\phi(\theta_1-\theta_2)-(\lambda_1-\lambda_2)=-2(\phi-\psi)(\theta_1-\theta_2),
		\]
		where we have used the relation $\lambda_k = 2\psi(\bar \theta - \theta_k)$.
		Therefore, we get
		\begin{multline}\label{eq:uniqueness:inequality1}
			2\phi Q \leq \bfJ(A^2,\theta_1,\lambda_1)-\bfJ(A^1,\theta_1,\lambda_1)+\bfJ(A^1,\theta_2,\lambda_2)-\bfJ(A^2,\theta_2,\lambda_2) \\
			= \lim_{T \to \infty} \frac{1}{T} \E \left[ \int_0^T \left( (h_1(X^2_t) - h_2(X^2_t) ) - (h_1(X^1_t) - h_2(X^1_t) ) \right) \dd t \right] \\
			= -2(\phi - \psi)(\theta_1 - \theta_2) \lim_{T\to \infty} \frac{1}{T} \E\left[ \int_0^T (X^2_t - X^1_t) \dd t \right] = 2(\phi-\psi)(\theta_1-\theta_2)^2,
		\end{multline}
		where in the last line we used the consistency conditions.
		On the other hand, by Jensen's inequality for every $T>0$ and using the consistency condition, we have
		\begin{multline}\label{eq:uniqueness:inequality2}
			Q=\varliminf_{T\to\infty}\frac{1}{T}\E\left[\int_0^T|X^1_t-X^2_t|^2 \dd t\right]\\
			\geq\varliminf_{T\to\infty}\left(\frac{1}{T}\E\left[\int_0^T(X^1_t-X^2_t) \dd t\right]\right)^2
			=\left(\lim_{T\to\infty}\frac{1}{T}\E\left[\int_0^T(X^1_t-X^2_t) \dd t\right]\right)^2
			=(\theta_1-\theta_2)^2.
		\end{multline}
		By putting together \eqref{eq:uniqueness:inequality1} and \eqref{eq:uniqueness:inequality2}, it follows that
		\[
		2(\phi-\psi)(\theta_1-\theta_2)^2\geq2\phi(\theta_1-\theta_2)^2,
		\]
		and therefore $-2\psi(\theta_1-\theta_2)^2 \geq 0$.
		Since $\psi>0$, we conclude that $\theta_1=\theta_2$.
		Hence, also $\lambda_1=\lambda_2$.
		For these fixed parameters, the reflected optimal control is unique by the uniqueness of the optimal thresholds and of the associated Skorohod reflection problem.
		Therefore $A^1=A^2$, and the stationary MFG solution is unique in the class of reflected controls.
	\end{proof}
	
	\subsection{From the stationary MFG to the stationary MFC}\label{app:stationary_MFC_problem}
	
	We derive some first-order optimality conditions for the optimality of the control $A^\star(\theta,\lambda)$, inspired by \cite{cannerozzi2026cooperation} (see also \cite[Proposition 16]{cannerozzi2026stationary}), which will be used to connect the solution of the stationary MFG to the solution of the stationary MFC problem. 
	
	\begin{proposition}\label{prop:first_order}
		Let $(\theta,\lambda) \in \R^2$ and let $A^\star(\theta,\lambda)$ be the optimal control for $\bfJ(A,\theta,\lambda)$ given by Proposition \ref{prop:hjb_frozen}-(\ref{prop:hjb_frozen:optimal_ctrl}).
		Then, for any $A \in \cB$ such that $\bfJ(A,\theta,\lambda)<+\infty$, it holds
		\begin{equation}\label{eq:optimality:first_order_condition}
			\varliminf_{T \to \infty} \frac{1}{T}\E\left[ \int_0^T \big( l_x(X^{A^\star}_t,\theta) - \lambda \big) (X^{A^\star}_t - X^{A}_t) \dd t  + K_+ (A^{\star,+}_T(\theta,\lambda) - A^{+}_T) + K_- (A^{\star,-}_T(\theta,\lambda) - A^{-}_T)   \right] \leq 0.
		\end{equation}
	\end{proposition}
	
	\begin{proof}
		To simplify the notation, we omit the dependence on $(\theta,\lambda)$ in the optimal control $A^\star$, and we set $h(x) \coloneqq l(x,\theta) +\lambda(\theta - x)$.
		For any $A \in \cB$ such that $\bfJ(A,\theta,\lambda)<+\infty$ and $\epsilon \in (0,\frac{1}{2}]$, set $A^{\epsilon} = \epsilon A + (1-\epsilon)A^{\star} = A^{\star} + \epsilon (A - A^{\star})$.
		Since $\cB$ is stable under convex combinations, $A^\epsilon\in\cB$.
		We define
		\begin{equation*}
			f(\epsilon,T) = \frac{1}{\epsilon}\bigg(  \frac{1}{T}\E\left[ \int_0^T h(X^{A^\epsilon}_t) \dd t + K_+ A^{\epsilon,+}_T + K_- A^{\epsilon,-}_T \right] - \frac{1}{T}\E\left[ \int_0^T h(X^{A^{\star}}_t) \dd t + K_+ A^{\star,+}_T + K_- A^{\star,-}_T\right]  \bigg).
		\end{equation*}
		By linearity of the state dynamics, $X^{A^\epsilon}_t=(1-\epsilon)X^{A^\star}_t+\epsilon X^A_t=X^{A^\star}_t+\epsilon(X^A_t-X^{A^\star}_t)$, $t \geq 0$.
		Moreover, since $h$ is quadratic with quadratic coefficient $\phi$, we have
		\begin{equation}\label{eq:first_order_condition:difference_h}
			h(X^{A^\epsilon}_t)-h(X^{A^\star}_t)=\epsilon h_x(X^{A^\star}_t)(X^A_t-X^{A^\star}_t)+\phi\epsilon^2|X^A_t-X^{A^\star}_t|^2.
		\end{equation}
		By minimality of the Jordan decomposition, for every $t\geq0$, $\P$-a.s., $A^{\epsilon,\pm}_t\leq \epsilon A^\pm_t+(1-\epsilon)A^{\star,\pm}_t = A^{\star,\pm}_t + \epsilon (A^\pm_t - A^{\star,\pm}_t)$. Then, this fact and identity \eqref{eq:first_order_condition:difference_h} imply that, for every $\epsilon\in(0,\frac12]$ and $T>0$, it holds
		\begin{multline}\label{eq:first_order_condition:bound_f}
			f(\epsilon,T) \leq\frac{1}{T}\E\left[\int_0^T h_x(X^{A^\star}_t)(X^A_t-X^{A^\star}_t) \dd t+K_+(A^+_T-A^{\star,+}_T)+K_-(A^-_T-A^{\star,-}_T)\right]\\
			+\phi\epsilon\frac{1}{T}\E\left[\int_0^T|X^A_t-X^{A^\star}_t|^2 \dd t\right].
		\end{multline}
		Since $h$ is quadratic with quadratic coefficient $\phi$, Young's inequality implies that, for fixed $(\theta,\lambda)$, there exists $C_{\theta,\lambda}>0$ such that $h(x)\geq\frac{\phi}{2}|x|^2-C_{\theta,\lambda}$ for any $x \in \R$.
		Therefore, since the intervention costs are nonnegative,
		\[
		\varlimsup_{T\to\infty}\frac{1}{T}\E\left[\int_0^T|X_t^A|^2 \dd t\right]\leq \frac{2}{\phi}\left( \bfJ(A,\theta,\lambda)+C_{\theta,\lambda} \right)<+\infty.
		\]
		Moreover, since $A^\star$ is a reflected policy between finite boundaries, $\varlimsup_{T\to\infty}\frac{1}{T}\E[\int_0^T|X^{A^\star}_t|^2dt]<+\infty$ as well, and therefore
		\begin{equation}\label{eq:first_order_condition:finite_limsup}
			\varlimsup_{T\to\infty}\frac{1}{T}\E\left[\int_0^T|X^A_t-X^{A^\star}_t|^2 \dd t\right]<+\infty.
		\end{equation}
		On the other hand, since the upper limit in $\bfJ(A^\star,\theta,\lambda)$ is an ordinary limit by Proposition \ref{prop:hjb_frozen}-\ref{prop:hjb_frozen:limsup},
		\[
		\varlimsup_{T\to\infty}f(\epsilon,T)=\frac{1}{\epsilon}\left(\bfJ(A^\epsilon,\theta,\lambda)-\bfJ(A^\star,\theta,\lambda)\right)\geq0,
		\]
		where the inequality follows from the optimality of $A^\star$.
		Combining this with \eqref{eq:first_order_condition:bound_f} yields
		\begin{align*}
			0 & \leq \varlimsup_{\epsilon \downarrow 0} \varlimsup_{T\to\infty} f(\epsilon,T) \\
			& \leq \varlimsup_{\epsilon \downarrow 0} \varlimsup_{T\to\infty} \bigg( \frac{1}{T}\E\bigg[\int_0^T h_x(X^{A^\star}_t)(X^A_t-X^{A^\star}_t) \dd t + K_+(A^+_T-A^{\star,+}_T)+K_-(A^-_T-A^{\star,-}_T)\bigg]\\
			& \quad +\epsilon \phi \frac{1}{T}\E\bigg[\int_0^T|X^A_t-X^{A^\star}_t|^2 \dd t\bigg] \bigg) \\
			& \leq \varlimsup_{T\to\infty}\frac{1}{T}\E\bigg[\int_0^T h_x(X^{A^\star}_t)(X^A_t-X^{A^\star}_t) \dd t + K_+(A^+_T-A^{\star,+}_T)+K_-(A^-_T-A^{\star,-}_T)\bigg],
		\end{align*}
		where the last inequality follows from the sublinearity of the upper limit and exploits \eqref{eq:first_order_condition:finite_limsup}.
		The established inequality is equivalent to \eqref{eq:optimality:first_order_condition}, and thus it concludes the proof.
	\end{proof}
	
	\begin{lemma}\label{lemma:mfc_finite_implies_aux_finite}
		Let $A\in\cB_{{\rm mfc}}$ be such that $\bfJ^{{\rm mfc}}(A)<+\infty$. Then $\bfJ(A,\theta,\lambda)<+\infty$ for every fixed $(\theta,\lambda)\in\R^2$.
	\end{lemma}
	\begin{proof}
		Let $\theta^A$ be the stationary mean generated by $A$. Since $x\mapsto l(x,\theta^A)$ is quadratic with positive quadratic coefficient, it is bounded from below by $\frac{\phi}{2}|x|^2-C$ for some constant $C>0$. Hence, $\bfJ^{{\rm mfc}}(A)<+\infty$ implies
		\[
		\varlimsup_{T\to\infty}\frac{1}{T}\E\left[\int_0^T |X^A_t|^2 \dd t+K_+A^+_T+K_-A^-_T\right]<+\infty.
		\]
		For fixed $(\theta,\lambda)$, the function $x\mapsto l(x,\theta)+\lambda(\theta-x)$ is again quadratic with positive quadratic coefficient. Hence it is bounded from above by $C(1+|x|^2)$ for some constant $C>0$, and the previous estimate gives $\bfJ(A,\theta,\lambda)<+\infty$.
	\end{proof}

	\begin{proof}[Proof of Theorem \ref{thm:mfg_to_mfc}]
		Let $A \in \cB_{{\rm mfc}}$ be any other control process. If $\bfJ^{{\rm mfc}}(A)=+\infty$, there is nothing to prove. Hence, let $\bfJ^{{\rm mfc}}(A)<+\infty$. By Lemma \ref{lemma:mfc_finite_implies_aux_finite}, we have $\bfJ(A,\theta^\star,\lambda^\star)<+\infty$, so that Proposition \ref{prop:first_order} can be applied with $(\theta,\lambda)=(\theta^\star,\lambda^\star)$.
		By Proposition \ref{prop:hjb_frozen}-\ref{prop:hjb_frozen:limsup}, the auxiliary cost $\bfJ(A^\star,\theta^\star,\lambda^\star)$ admits an ordinary limit. Moreover, since $A^\star\in\cB_{{\rm mfc}}$ and $\theta^\star$ is the stationary mean generated by $A^\star$,
		\[
		\lim_{T\to\infty}\frac{1}{T}\E\left[\int_0^T\lambda^\star(\theta^\star-X_t^{A^\star}) \dd t\right]=0.
		\]
		Hence, the cost defining $\bfJ^{{\rm mfc}}(A^\star)$ admits an ordinary limit as well.
		Using $\lim_{T\to\infty}a_T-\varlimsup_{T\to\infty}b_T = \varliminf_{T\to\infty}(a_T-b_T)$, and relying on the convexity of $(x,\theta) \mapsto l(x,\theta)$, we get
		\begin{equation}\label{eq:thm:sufficient_condition:bound1}
			\begin{aligned}
				& \bfJ^{{\rm mfc}} (A^\star) - \bfJ^{{\rm mfc}}(A) \leq \varliminf_{T \to \infty} \frac{1}{T}\E\bigg[ \int_0^T \big( l_x(X^{A^\star}_t,\theta^{\star})(X^{A^\star}_t - X^{A}_t) \\
				& \;\; + l_\theta(X^{A^\star}_t,\theta^{\star})(\theta^{\star} - \theta^{A}) \big)\dd t + K_+ (A^{\star,+}_T - A^{+}_T) + K_- (A^{\star,-}_T - A^{-}_T)   \bigg] \\
				& \leq \varliminf_{T \to \infty} \frac{1}{T}\E\left[ \int_0^T l_x(X^{A^\star}_t,\theta^{\star})(X^{A^\star}_t - X^{A}_t) \dd t  + K_+ (A^{\star,+}_T - A^{+}_T) + K_- (A^{\star,-}_T - A^{-}_T)   \right] \\
				& \;\; + \varlimsup_{T \to \infty}\frac{1}{T}\E\left[\int_0^T l_\theta(X^{A^\star}_t,\theta^{\star}) \dd t \right](\theta^{\star} - \theta^{A}).
			\end{aligned}
		\end{equation}
		By exploiting \eqref{mfg:cons:lambda}, we have
		\begin{equation*}
			\lim_{T \to \infty}\frac{1}{T}\E\left[ \int_0^T l_\theta(X^{A^\star}_t,\theta^{\star}) \dd t \right] = -\lambda^\star.
		\end{equation*}
		As $A^\star$ and $A$ belong to $\cB_{{\rm mfc}}$, the processes $X^{A^\star}$ and $X^A$ are ergodic, so that we have
		\begin{equation}\label{eq:cons_lambda:implication}
			\varlimsup_{T \to \infty}\frac{1}{T}\E\left[\int_0^T l_\theta(X^{A^\star}_t,\theta^{\star}) \dd t \right](\theta^{\star} - \theta^{A}) = -\lambda^\star \lim_{T \to \infty}\frac{1}{T}\E\left[\int_0^T (X^{A^\star}_t-X^{A}_t) \dd t \right].
		\end{equation}
		By plugging \eqref{eq:cons_lambda:implication} in the last line of \eqref{eq:thm:sufficient_condition:bound1}, we deduce
		\begin{equation}\label{eq:thm:sufficient_condition:bound2}
			\begin{aligned}
				& \bfJ^{{\rm mfc}}(A^\star) - \bfJ^{{\rm mfc}}(A) \leq \varliminf_{T \to \infty} \frac{1}{T}\E\left[ \int_0^T l_x(X^{A^\star}_t,\theta^{\star})(X^{A^\star}_t - X^{A}_t) \dd t  + K_+ (A^{\star,+}_T - A^{+}_T) + K_- (A^{\star,-}_T - A^{-}_T)   \right] \\
				& \;\; - \lim_{T \to \infty}\frac{1}{T}\E\left[\int_0^T \lambda^\star (X^{A^\star}_t-X^{A}_t) \dd t \right] \\
				& \leq \varliminf_{T \to \infty} \frac{1}{T}\E\left[ \int_0^T \big( l_x(X^{A^\star}_t,\theta^{\star}) - \lambda^\star\big) (X^{A^\star}_t - X^{A}_t)\dd t  + K_+ (A^{\star,+}_T - A^{+}_T) + K_- (A^{\star,-}_T - A^{-}_T)   \right].
			\end{aligned}
		\end{equation}
		By \eqref{eq:optimality:first_order_condition}, the right-hand side of \eqref{eq:thm:sufficient_condition:bound2} is always non-positive, which implies that $A^\star$ is optimal.
	\end{proof}

\subsection{Numerical method}
	
Proposition~\ref{prop:main_fixed_point_characterization} reduces the computation of the optimal MSR policy to three nested tasks, namely (1) solving a frozen obstacle problem, (2) computing the invariant distribution generated by the resulting reflecting bands, and (3) imposing the scalar consistency condition. Our implementation follows this construction directly. In particular, it does not impose an ordering of the four regime-dependent barriers. This is important because their relative ordering may change with the parameters and because degenerate configurations, in which two barriers coincide, may arise. For a fixed candidate mean $\theta$, set $\lambda(\theta)=2\psi(\bar\theta-\theta),$ and $c^\theta(x):=l_x(x,\theta)-\lambda(\theta) = \rho+2\phi x+2(\psi-\phi)\theta-2\psi\bar\theta.$
	
\paragraph{(1) Frozen obstacle problem.}
	We truncate the state space to a sufficiently large interval $[\underline x,\overline x]$ and introduce a grid
	$\underline x=x_1<\cdots<x_N=\overline x$. For $i\in\cE$ and $j\neq i$, define
	\[
	\cH_i^\theta(U)(x)
	:=-\delta_iU_i(x)-\delta_i x\,\partial_xU_i(x)
	+\frac{\sigma^2}{2}\partial_{xx}U_i(x)
	+\gamma_i\bigl(U_j(x)-U_i(x)\bigr)+c^\theta(x),
	\]
	where $U_i(x):=U(x,i)$. The obstacle problem in
	\eqref{eq:main_thm:free_boundary_problem} can be written as
	\[
	\max\Big\{
	\min\big\{\cH_i^\theta(U)(x),K_--U_i(x)\big\},
	-K_+-U_i(x)
	\Big\}=0,
	\qquad i\in\cE.
	\]
	We discretize the differential operators by finite differences and solve the resulting coupled discrete variational inequality by policy iteration. At every iteration, each grid point is assigned to the lower-contact, continuation, or upper-contact region. The obstacle value is imposed at contact points, whereas the discretized equation $\cH_i^\theta(U)=0$ is imposed at continuation points. The resulting linear system for $(U_1,U_2)$ is solved, the active sets are updated, and the procedure is repeated until both the active sets and the projected residual are stable. At the endpoints of the truncated domain we impose the obstacle values $U_i(\underline x)=-K_+$ and $U_i(\overline x)=K_-$. The domain is chosen sufficiently wide that enlarging it does not alter the computed barriers at the reported precision.
	
The thresholds are recovered from the contact sets by interpolation:
\[
a_+(i,\theta) :=\sup\{x:U_i^\theta(x)=-K_+\}, \qquad a_-(i,\theta) :=\inf\{x:U_i^\theta(x)=K_-\}.
\]
No ordering is imposed across regimes; the configuration of the four barriers is determined only after the obstacle problem has been solved.
	
\paragraph{Invariant distribution and stationary-mean map.}
Given the barriers, we compute the regime-wise stationary CDFs $\mu_i^\theta(x):=\mu^\theta(x,i)$ by applying finite differences to the coupled system
\eqref{eq:regime_wise_CDF}, with
\[
\mu_i^\theta(x)=0\quad\text{for }x\leq a_+(i,\theta), \qquad \mu_i^\theta(x)=\pi_i\quad\text{for }x\geq a_-(i,\theta).
\]
The stationary mean generated by the candidate bands is then evaluated through the compact representation in \eqref{eq:fixed_point_map}:
\[ 
\Phi(\theta) = \sum_{i=1}^2 \Big(\pi_i a_-(i,\theta)  - \int_{a_+(i,\theta)}^{a_-(i,\theta)}\mu_i^\theta(x)\,\dd x \Big).
\]
Thus, one evaluation of $\Phi$ consists of one obstacle solve followed by one stationary-distribution solve.
	
\paragraph{(3) Outer fixed-point search.}
	We define $F(\theta):=\Phi(\theta)-\theta$ and first evaluate $F$ on a coarse grid of candidate means. Once a sign change has been located, the unique root is refined by bisection. If the initial search interval does not contain a sign change, it is enlarged until one is found. The existence and uniqueness results in Section~\ref{subsec:fixed_point_auxiliary} imply that the exact map has a single fixed point. The numerical output is
$\theta^\star,$	$\lambda^\star=\lambda(\theta^\star),$ and the thresholds $a_\pm^\star(i)=a_\pm(i,\theta^\star),$ for $i\in\cE.$
	
\paragraph{Simulation and numerical checks.}
	Sample paths are generated by an Euler scheme with Markov switching. At each time step, the regime and the uncontrolled state increment are updated, after which the state is projected onto the reflecting interval associated with the updated regime. The upward and downward projection amounts are recorded as the increments of $A^+$ and $A^-$. This procedure also captures the instantaneous adjustment that may be triggered by a regime switch.

	\bibliographystyle{abbrv}

\end{document}